\documentclass[final]{siamltex}
\usepackage{enumerate}
\usepackage{amsmath}
\usepackage{amssymb}
\usepackage{amsfonts}

\usepackage{mathrsfs}
\newcommand\cF{{\cal F}}
\newcommand\cY{{\cal Y}}
\newcommand\bE{\mathbf{E}}
\newcommand\bR{\mathbb{R}}
\newcommand\bP{\mathbb{P}}
\newcommand\be{\begin{equation}}
\newcommand\ee{\end{equation}}

\newtheorem{assumption}{Assumption}
\title{Jump-Diffusion Risk-Sensitive Asset Management I: Diffusion Factor Model}
\author{Mark Davis\thanks{Department of Mathematics, Imperial College London, London SW7 2AZ, England, Email: mark.davis@imperial.ac.uk} \and S\'ebastien Lleo\thanks{Finance Department, Reims Management School, 59 rue Pierre Taittinger, 51100 Reims, France, Email: sebastien.lleo@reims-ms.fr}}
\date{September 26, 2010}

\begin{document}
\bibliographystyle{plain}
\maketitle

\begin{abstract}
        This paper considers a portfolio optimization problem in which asset prices are represented by SDEs driven by Brownian motion and a Poisson random measure,
with drifts that are functions of an auxiliary diffusion factor process. The criterion, following earlier work by Bielecki, Pliska, Nagai and others, is risk-sensitive optimization (equivalent to maximizing the expected growth rate subject to a constraint on variance.) By using a change of measure technique introduced by Kuroda and Nagai we show that the problem reduces to solving a certain stochastic control problem in the factor process, which has no jumps. The main result of the paper is to show that the risk-sensitive jump diffusion problem can be fully characterized in terms of a parabolic Hamilton-Jacobi-Bellman PDE rather than a PIDE, and that this PDE admits a classical $(C^{1,2})$ solution.
\end{abstract}


\section{Introduction}
In this article, we consider a finite time jump-diffusion version of the risk sensitive asset management problem of Bielecki and Pliska~\cite{bipl99}. Fundamentally, our main result is to show that the resulting stochastic control problem can be fully characterized by a parabolic Hamilton-Jacobi-Bellman PDE rather than a PIDE, and that this PDE admits a classical $(C^{1,2})$ solution.  
\\

Risk-sensitive control is a generalization of classical stochastic control in which the degree of risk aversion or risk tolerance of the optimizing agent is explicitly parameterized in the objective criterion and influences directly the outcome of the optimization. In risk-sensitive control, the decision maker's objective is to select a control
policy $h(t)$ to maximize the criterion
\begin{equation}\label{eq_criterion_J}
    J(x,t,h;\theta) := -\frac{1}{\theta}\ln\mathbf{E}\left[e^{-\theta F(x,h)} \right]
\end{equation}
where $t$ is the time, $x$ is the state variable, $F$ is a given reward function, and the risk sensitivity $\theta>0$ is an exogenous parameter representing the decision maker's degree of risk aversion. A Taylor expansion of the previous expression around $\theta = 0$
evidences the vital role played by the risk sensitivity parameter:
\begin{equation}\label{eq_Taylor_RSC}
    J(x,t,h;\theta)
    = \mathbf{E}\left[F(x,t,h)\right]
    - \frac{\theta}{2}\mathbf{Var}\left[F(x,t,h)\right]
    + O(\theta^2)
\end{equation}
This criterion amounts to maximizing $\mathbf{E}\left[F(x,t,h)\right]$ subject to a penalty for variance.  For a general reference, see Whittle~\cite{wh90}. Much of the recent literature concerns the infinite time horizon problem:
\begin{equation}\label{eq_criterion_J_infinitetime}
    J_{\infty}(x, h;\theta) := \liminf_{t \to \infty} -\frac{1}{\theta} t^{-1} \ln\mathbf{E}\left[e^{-\theta f(x,h)} \right]
\end{equation}
This is interesting from a theoretical perspective, but is not applicable to practical asset management because of the non-uniqueness of controls. Optimality in this sense is a `tail property': if $h^*(t)$ is optimal, then so is $\tilde{h}(t) = h^*(t) 1_{t>T}+h(t) 1_{t \leq T}$ for any arbitrary process $h(t)$ and time $T>0$. Of course, near-term decisions are the ones that are of primary importance to investment managers.

In the past decade, the applications of risk-sensitive control to asset management have flourished. Risk-sensitive control was first applied to solve financial problems by Lefebvre and Montulet~\cite{lemo94} in a corporate finance context and by Fleming~\cite{fl95} in a portfolio selection context. However, Bielecki and Pliska~\cite{bipl99} were the first to apply the continuous time risk-sensitive control as a practical tool that could be used to solve `real world' portfolio selection problems. They considered a long-term asset allocation problem and proposed the logarithm of the investor's wealth as a reward function, so that the investor's objective is to maximize the risk-sensitive (log) return of his/her portfolio or alternatively to maximize a function of the power utility (HARA) of terminal wealth. They derived the optimal control and solved the associated Hamilton-Jacobi-Bellman (HJB) PDE under the restrictive assumption that the asset and factor noise are uncorrelated. This assumption is unrealistic and it was later relaxed (see~\cite{bipl04}). The contribution of Bielecki and Pliska to the field is immense: they studied the economic properties of the risk-sensitive asset management criterion (see~\cite{bipl03}), extended the asset management model into an intertemporal CAPM (\cite{bipl04}), worked on transaction costs (\cite{bipl00}), numerical methods (\cite{bihhpl02}) and considered factors driven by a CIR model (\cite{biplsh05}).  A major contribution to the mathematical theory was made by Kuroda and Nagai~\cite{kuna02} who introduced an elegant solution method based on a change of measure argument which transforms the risk sensitive control problem in a linear exponential of quadratic regulator. They solved the associated HJB PDE over a finite time horizon and then studied the properties of the ergodic HJB PDE. Recently, Davis and Lleo~\cite{dall_RSBench} applied this change of measure technique to solve, at a finite and an infinite horizon, a benchmarked investment problem in which an investor selects an asset allocation to outperform a given financial benchmark. The problem we consider is also related to the vast literature on HARA utility maximization that has flourished in the past 50 years. This literature includes a number of references related to risk-sensitive control, such as works by Fleming and Sheu (\cite{flsh99}, \cite{flsh00} and~\cite{flsh02}) or Hansen and Sargent~\cite{hasa08} in the context of robust control.
\\

Risk-sensitive asset management theory was originally set in a world of diffusion dynamics where randomness is modelled using correlated Brownian motions. To our knowledge, the only attempt to extend the risk-sensitive asset management theory from a diffusion to a jump diffusion setting was made by Wan~\cite{wa06} who briefly sketched a jump-diffusion extension of Bielecki and Pliska's~\cite{bipl99} original infinite horizon risk-sensitive asset management model. Wan's treatment is however restrictive as it only considers a single Poisson process-driven jump per asset and assumes that the underlying valuation factor risks and asset risks are uncorrelated. Our paper addresses these two limitations. The setting of our control problem, which takes place within a finite time horizon, allows for both infinite activity jumps in asset prices and for a correlation structure between factor risks and asset risks. To solve this control problem we extend Kuroda and Nagai's powerful change of measure technique to account for the jumps. One of the difficulties we face in extending this technique is proving that the optimal control is admissible as this requires showing that the Dol\'eans exponential~\eqref{eq_JDRSAM_Doleansexp_chi} associated with this control is a martingale. In a pure diffusion setting, this would follow easily from the Kamazaki condition or the Novikov condition. However, when the Dol\'eans exponential does not have continuous path, as is the case in a jump diffusion setting, proving that it is indeed a martingale is more difficult as only weaker partial results exist. This question is addressed in Appendix \ref{sec_Admissibility} of the present paper.

In this paper the asset price processes are modelled as jump-diffusions whose growth rates are functions of an auxiliary `factor' process $X(t)$ which satisfies a linear diffusion SDE. Our main result is that the risk-sensitive jump-diffusion asset management problem is equivalent to an optimal control problem for a \emph{diffusion} process (no jumps) and that the HJB equation for the latter admits a unique classical $C^{1,2}\left([0,T)\times\mathbb{R}^n\right)$ solution. Showing the existence and uniqueness of a solution to a risk-sensitive control problem can prove difficult even in a pure diffusion setting. For example, Bensoussan, Freshe and Nagai~\cite{befrna98} had to constrain the behaviour of the Hamiltonian in order to prove existence of a classical solution. Still in a pure diffusion setting, Fleming and Soner (see V.9 in~\cite{flso06}) proved that the value function is a continuous viscosity solution of the associated Hamilton Jacobi Bellman Partial Differential Equation (HJB PDE) but had to assume boundedness of all coefficients and of the derivatives of the reward function. No such strong condition is required to solve the jump diffusion problem considered in this article. In fact all our assumptions arise naturally from the structure of the risk-sensitive asset management problem. Uniqueness follows from a classical verification argument while the proof of existence relies on a policy improvement algorithm and on the properties of linear parabolic PDEs.
\\

The paper is organized as follows. We first introduce the general setting of the model in section 2 and define the class of random Poisson measures which will be used to model the jump component of the asset dynamics. In Section 3, we formulate the jump-diffusion control problem and introduce the change of measure argument of Kuroda and Nagai~\cite{kuna02}. In a pure diffusion case, this is enough to transform the problem into a standard Linear Exponential of Quadratic Regulator. In our jump-diffusion setting, the change of measure simplifies the problem by associating the HJB PDE given in Section \ref{sec_Optim}, rather than the expected Partial Integro-Differential Equation containing non-local terms, to the value function. It is striking that an optimal control problem for a jump-diffusion model has a solution that is characterized in terms of a HJB PDE and not a HJB PIDE\footnote{See \O ksendal and Sulem~\cite{oksu05} for a treatment of jump-diffusion control problems.}.
\\

Our main result is Theorem~\ref{theo_JDRSAM_diffstate_main_result} in Section~\ref{sec_Mainresult}. The proof depends on various technical arguments which are given in Sections 5 to 7. In section 5, we show the existence of a unique optimal control before addressing two key questions in Section 6. First, the admissibility of the optimal control is no longer \emph{a priori} guaranteed because the Dol\'eans exponential defining the Radon-Nikodym derivative does not have continuous paths. This point is addressed in Propositions 6.3 and 6.4. Second, the Risk-Sensitive Hamilton-Jacobi-Bellman Partial Differential Equation (RS HJB PDE)  contains a jump-induced control-dependent integral term: it is no longer possible to find an analytical solution and the existence of a strong, classical solution is no longer guaranteed. However, should we be able to prove the existence of a classical $C^{1,2}$ solution to the RS HJB PDE, then we can prove uniqueness and resolve the control problem using a straightforward verification theorem, presented in Theorem 6.1 and Corollary 6.2 in Section 6. 
\\

In Section 7 we address existence and regularity of solutions to the RS HJB PDE. We show, in Theorem 7.2 and Corollary 7.4, that the risk-sensitive jump-diffusion control problem we consider admits a unique classical $C^{1,2}\left([0,T)\times\mathbb{R}^n\right)$ solution. Showing the existence and uniqueness of a solution to a risk-sensitive control problem can prove difficult even in a pure diffusion setting. For example, Bensoussan, Frehse and Nagai~\cite{befrna98} had to constrain the behaviour of the Hamiltonian in their finite time horizon problem to prove existence of a classical solution. Still in a pure diffusion setting and over a finite time horizon, Fleming and Soner (see V.9 in~\cite{flso06}) proved that the value function is a continuous viscosity solution of the associated Hamilton Jacobi Bellman Partial Differential Equation (HJB PDE) but had to assume boundedness of all coefficients and of the derivatives of the reward function. No such strong condition is required to solve the jump diffusion problem considered in this article. In fact all our assumptions arise naturally from the structure of the risk-sensitive asset management problem.  We obtain our result by applying an approximation in policy space in a two-step process: first, we show existence on a  bounded region and then extend to the unbounded state space. With this result in hand, we have all the ingredients needed for the proof of Theorem~\ref{theo_JDRSAM_diffstate_main_result}. 
\\

Up to this point, we have assumed that the factor process $X(t)$ is directly observed by the controller, and therefore represents real economic factors: GDP growth, inflation, the S\&P500 index, etc. We may however wish to use $X(t)$ as an abstract latent factor, introduced to model volatility of returns, in which case only the prices, and not $X(t)$, will be observed. In our final Section 8, we note that this problem, once adequately reformulated, can be solved using a classical Kalman filter, as in \cite{nape02}, as the jump noise is absent from the dynamics of $X(t)$. While this is from a technical point of view a simple observation, it greatly enhances the applicability of our results.
\\

In a companion paper \cite{dall10} we consider the case in which there are jumps in both the price and factor processes. There the measure change technique does not remove the jumps and the argument is substantially different.

%
%

\section{Analytical Setting}\label{as}
\subsection{Overview}\label{def_JDRSAM_setting}
The growth rates of the assets are assumed to depend on $n$ factors $X_1(t), \ldots, X_n(t)$ which follow the dynamics given in equation~\eqref{eq_FactorProcess_diffusion} below. As in Kuroda and Nagai's asset-only model, the assets market comprises $m$ risky securities $S_i, \; i=1,\ldots m$. In contrast to Kuroda and Nagai, we assume that the money market account process, $S_0$, is an affine function of the valuation factors, which enables us to easily model a stochastic short term rate. Let $M := n+m$. Throughout, we will assume that $m>n$; this is needed in connection with the `zero-beta' policies introduced in Section \ref{zb}.
\\

Let $(\Omega, \left\{ \mathcal{F}_{t} \right\}, \mathcal{F},
\mathbb{P})$ be the underlying probability space. On this space is
defined an $\mathbb{R}^M$-valued
$\left(\mathcal{F}_t\right)$-Brownian motion $W(t)$ with components
$W_k(t)$, $k=1,\ldots,M$. Moreover, let $N$ be a $\left(\mathcal{F}_t\right)$-Poisson point process on $(0,\infty)\times\mathbf{Z}$, independent of $W(t)$, where $(\mathbf{Z},\mathcal{B}_{\mathbf{Z}})$ is a given Borel space\footnote{$\mathbf{Z}$ is a Polish space and $\mathcal{B}_{\mathbf{Z}}$ is the Borel $\sigma$-field. See Ikeda and Watanabe~\cite{ikwa81} for a formal definition of the Poisson point process}. Define
\begin{equation}\label{def_JDRSAM_ZFrak_set}
    \mathfrak{Z} := \left\{ U \in \mathcal{B}(\mathbf{Z}), \mathbb{E} \left[N(t,U)\right] < \infty \; \forall t\right\}
\end{equation}
Finally, for notational convenience, we fix throughout the paper a set $\mathbf{Z}_0 \in \mathcal{B}_{\mathbf{Z}}$ such that
$\nu(\mathbf{Z} \backslash \mathbf{Z}_0)<\infty$ and define 
\begin{eqnarray}
    &&\bar{N}(dt,dz)
                                                \label{nbar}\\
    &=& \left\{ \begin{array}{ll}
        N(dt,dz) - \hat{N}(dt,dz) = N(dt,dz) - \nu(dz)dt =: \tilde{N}(dt,dz)    &   \textrm{if } z \in \mathbf{Z}_0     \\
        N(dt,dz)                      &   \textrm{if } z \in \mathbf{Z} \backslash \mathbf{Z}_0       \\
    \end{array}\right.
                                                \nonumber
\end{eqnarray}

\subsection{Factor Dynamics}\label{Chapter3_JDRSAM_theory_factordynamics}
The dynamics of the $n$ factors are expressed by the affine diffusion equation
\begin{equation}\label{eq_FactorProcess_diffusion}
    dX(t) = (b + BX(t))dt + \Lambda dW(t),
    \qquad X(0) = x
\end{equation}
where $X(t)$ is the $\mathbb{R}^{n}$-valued factor process with components $X_{j}(t)$ and $b \in \mathbb{R}^n$, $B \in \mathbb{R}^{n\times n}$ and $\Lambda \in \mathbb{R}^{n\times M}$.
\\

\subsection{Asset Market Dynamics}\label{Chapter3_JDRSAM_theory_assetdynamics}
Let $S_0$ denote the wealth invested in the money market account with dynamics given by the equation:
\begin{equation}\label{eq_JDRSAM_BankAccount}
    \frac{dS_{0}(t)}{S_{0}(t)} = \left(a_0 + A_0^T X(t)\right)dt,
    \qquad S_0(0) = s_0
\end{equation}
where $a_0 \in \mathbb{R}$ is a scalar constant, $A_0 \in
\mathbb{R}^{n}$ is a $n$-element column vector and throughout the paper $x^T$ denotes the transpose of the matrix or vector $x$.
\\

Let $S_{i}(t)$ denote the price at time $t$ of the $i$th security, with $i = 1,
\ldots, m$. The dynamics of risky security $i$ can be expressed as:
\begin{eqnarray}\label{eq_SecurityProcess}
    \frac{dS_{i}(t)}{S_{i}(t^{-})} &=&
        (a + AX(t))_{i}dt
        + \sum_{k=1}^{N} \sigma_{ik} dW_{k}(t)
        + \int_{\mathbf{Z}} \gamma_i(z)\bar{N}(dt,dz),
                                                \nonumber\\
    && S_i(0) = s_i,
    \quad i = 1, \ldots, m
\end{eqnarray}
where $a \in \mathbb{R}^m$, $A \in \mathbb{R}^{m \times n }$,
$\Sigma := \left[ \sigma_{ij} \right], \; i = 1, \ldots, m, \; j =
1, \ldots, M$ and $\gamma(z) \in \mathbb{R}^m$ satisfying Assumption~\ref{as_assetjumps_upanddown_1}:
\\

\begin{assumption}\label{as_assetjumps_upanddown_1}
$\gamma(z) \in \mathbb{R}^m$ satisfies
\begin{eqnarray}
        -1 \leq \gamma_{i}^{min} \leq \gamma_{i}(z) \leq \gamma_{i}^{max} < +\infty
        , \qquad i =1,\ldots,m
                                                                                        \nonumber
\end{eqnarray}and
\begin{eqnarray}
        -1 \leq \gamma_{i}^{min} < 0 < \gamma_{i}^{max} < +\infty
        , \qquad i =1,\ldots,m
                                                                                        \nonumber
\end{eqnarray}
for $i = 1, \ldots, m$. Furthermore, define
\begin{equation}
    \mathbf{S} := \mathrm{supp}(\nu) \in \mathcal{B}_{\textbf{Z}}
                                                \nonumber
\end{equation}
and
\begin{equation}
    \tilde{\mathbf{S}}
    :=  \mathrm{supp}(\nu \circ \gamma^{-1})
    \in \mathcal{B}\left(\mathbb{R}^m\right)
                                                \nonumber
\end{equation}
where $\mathrm{supp}(\cdot)$ denotes the measure's support, then we assume that $\prod_{i=1}^{m} [\gamma_{i}^{min}, \gamma_{i}^{max}]$ is the smallest closed hypercube containing $\tilde{\mathbf{S}}$.
\\

In addition, the vector-valued function $\gamma(z)$ satisfies:
\begin{equation}\label{as_assetjumps_gamma_integrable}
    \int_{\mathbf{Z}_0} \lvert\gamma(z)\rvert^2 \nu(dz) < \infty
\end{equation}
\end{assumption}
\\

Note that Assumption~\ref{as_assetjumps_upanddown_1} implies that each asset has, with positive probability, both upward and downward jumps. As will become evident in Section~\ref{sec_Change_of_Measure}, the effect of this assumption is to bound the space of controls. Relation~\eqref{as_assetjumps_gamma_integrable} is a standard condition. 
\\

Define the set $\mathcal{J}$ as
\begin{equation}\label{def_JDRSAM_setJ}
    \mathcal{J} := \left\{h \in \mathbb{R}^m:  -1-h^T \psi < 0 \quad \forall \psi \in \tilde{\mathbf{S}}\right\}
\end{equation}
and let $\overline{\mathcal{J}}$ be the closure of $\mathcal{J}$. For a given $z$, the equation $h^T \gamma(z)   = -1$ describes a
hyperplane in $\mathbb{R}^m$. $\mathcal{J}$ is a bounded open convex subset of $\mathbb{R}^m$. 
\\

\subsection{Portfolio Dynamics}\label{Chapter3_JDRSAM_theory_portfoliodynamics}
 We will need the following assumptions:
\\
\begin{assumption}\label{as_JDRSAM_sigposdef} $\Sigma\Sigma^T  >0$
\end{assumption}
\\

The effect of this assumption is to prevent redundant assets. For example, we will not able to model in our investment market a share and an option or futures on that share. However, this assumption leaves us free to model a wide range of assets such as shares, bonds and commodities products as well as related indices.
\\

\begin{assumption}\label{as_JDRSAM_lamposdef} $\Lambda\Lambda^T  >0.$
\end{assumption}
\\

Let $\mathcal{G}_t := \sigma((S(s), X(s)), 0 \leq s \leq t)$ be the
sigma-field generated by the security and factor processes up to
time $t$. An \textit{investment strategy} or \textit{control process} is an $\mathbb{R}^m$-valued process with the interpretation that $h_i(t)$ is the fraction of current portfolio value invested in the $i$th asset, $i=1,\ldots,m$. The fraction invested in the money market account is then $h_0(t) = 1 - \sum_{i=1}^{m} h_{i}(t)$.
\\

\begin{definition}\label{def_JDRSAM_controlprocess_h}
    An $\mathbb{R}^m$-valued control process $h(t)$ is in class $\mathcal{H}$ if the
    following conditions are satisfied:
    \begin{enumerate}
        \item $h(t)$ is progressively measurable with respect to
        $\left\{ \mathcal{B}([0,t]) \otimes \mathcal{G}_t\right\}_{t \geq
        0}$ and is c\`adl\`ag;

        \item $\bP\left(\int_{0}^{t} \left| h(s) \right|^2 ds < +\infty \right)
        =1, \quad \forall t>0$;

        \item $h^T (t)\gamma(z) > -1, \quad \forall t >0, z \in \mathbf{Z}$, a.s. $d\nu$.
    \end{enumerate}
\end{definition}

Define the set $\mathcal{K}$ as
\begin{equation}\label{def_JDRSAM_setmathcalK}
    \mathcal{K} := \left\{h \in \mathcal{H}: h(t,\omega) \in \mathcal{J} \quad \textrm{a.e. } (dt\times d\mathbb{P} \right\}
\end{equation}

\begin{lemma}
        Under Assumption~\ref{as_assetjumps_upanddown_1}, a control process $h(t)$ satisfying condition 3 in Definition~\ref{def_JDRSAM_controlprocess_h} is bounded.
\end{lemma}

\begin{proof}
        The proof of this result is immediate.
\end{proof}

\begin{definition}\label{def_JDRSAM_admissible_A}
    A control process $h$ is in class $\mathcal{A}(T)$ if the
    following conditions are satisfied:
    \begin{enumerate}
        \item $h \in \mathcal{H}$;

                        \item $\mathbf{E} \chi_T^h= 1$ where $\chi_t^h$, $t \in (0,T]$, is the Dol\'eans exponential defined as
\begin{eqnarray}\label{eq_JDRSAM_Doleansexp_chi}
    \chi_t^h
    &:=& \exp \left\{ -\theta \int_{0}^{t} h(s)^T \Sigma dW_s
    -\frac{1}{2} \theta^2 \int_{0}^{t} h(s)^T \Sigma\Sigma^T h(s) ds            \right.
                                                            \nonumber\\
   &&   \left.
        +\int_{0}^{t} \int_{\mathbf{Z}} \ln\left(1-G(z,h(s))\right) \tilde{N}(ds,dz)
            \right.
                                                            \nonumber\\
   &&   \left.
        +\int_{0}^{t} \int_{\mathbf{Z}} \left\{\ln\left(1-G(z,h(s))\right)+G(z,h(s))\right\}\nu(dz)ds
    \right\},
                                                                                                                                                                \nonumber\\
\end{eqnarray}
and
\begin{eqnarray}
    G(z,h) &=&
        1-\left(1+h^T \gamma(z)\right)^{-\theta}
\end{eqnarray}

\end{enumerate}
 We  say that a control process $h$ is \emph{admissible} if $h\in \mathcal{A}(T)$.
\end{definition}

The proportion invested in the money market account is $h_0(t)=1-\sum_{i=1}^{m} h_i(t)$. Taking this budget equation into consideration, the wealth, $V(t)$ of the investor in response to an investment strategy $h(t) \in \mathcal{H}$, follows the dynamics
\begin{eqnarray}
   \frac{dV(t)}{V(t^-)}
    &=& \left(a_0 + A_0^T X(t)\right)dt
            + h^T (t)\left(a-a_0\mathbf{1}
            +\left(A-\mathbf{1}A_0^T \right)X(t)\right)dt
                                                    \nonumber\\
    &&
            + h^T (t)\Sigma dW_t
            + \int_{\mathbf{Z}}h^T (t)\gamma(z)\bar{N}(dt,dz)
                                        \nonumber
\end{eqnarray}
where $\mathbf{1} \in \mathbf{R}^m$ denotes the $m$-element unit column vector and with $V(0) = v$. Defining $\hat{a} := a - a_0\mathbf{1}$ and $\hat{A} := A -
\mathbf{1}A_0^T $, we can express the portfolio dynamics as
\begin{eqnarray}\label{eq_JDRSAM_V_dynamics}
    \frac{dV(t)}{V(t^-)}
        = \left(a_0 + A_0^T X(t)\right)dt
        + h^T (t)\left(\hat{a}+\hat{A}X(t)\right)dt
        + h^T (t)\Sigma dW_t
        + \int_{\mathbf{Z}}h^T (t)\gamma(z)\bar{N}(dt,dz)
                                        \nonumber\\
\end{eqnarray}
with initial endowment $V(0) = 0$.
\\

%
%

\section{Problem Setup}
\subsection{Optimization Criterion}
We will assume that the objective of the investor is to maximize the risk adjusted growth of his/her portfolio of assets over a finite time horizon. In this context, the objective of the risk-sensitive management problem is to find $h^*\in \mathcal{A}(T)$ that maximizes the control criterion
\begin{equation}\label{eq_JDRSAM_criterion_J}
    J(x,t,h;\theta) := -\frac{1}{\theta}\ln\mathbf{E}\left[e^{-\theta \ln V(t,x,h)} \right]
\end{equation}

By It\^o's Lemma, the log of the portfolio value in response to a strategy $h$ is
\begin{eqnarray}\label{eq_JDRSAM_Vt}
    \ln V(t)
        &=&\ln v + \int_{0}^{t} \left(a_0 + A_0^T X(s)\right)
        +   h(s)^T \left(\hat{a}+\hat{A}X(s)\right)ds
        -   \frac{1}{2} \int_{0}^{t} h(s)^T \Sigma\Sigma^T h(s) ds
                                                                      \nonumber \\
        &&  +\int_{0}^{t} h(s)^T \Sigma dW(s)
                                                                            \nonumber \\
        &&  +\int_{0}^{t}\int_{\mathbf{Z}_0}
            \left\{
            \ln\left(1+h(s)^T \gamma(z)\right)-h(s)^T \gamma(z)
            \right\}\nu(dz)ds
                                                                            \nonumber \\
        &&  +\int_{0}^{t} \int_{\mathbf{Z}} \ln\left(1+h(s)^T \gamma(z)\right) \bar{N}(ds,dz)
\end{eqnarray}
Hence,
\begin{eqnarray}\label{eq_JDRSAM_eminthetaVt}
    e^{-\theta \ln V(t)} &=&  v^{-\theta}
        \exp \left\{ \theta \int_{0}^{t} g(X_s,h(s)) ds \right\} \chi_t^h
\end{eqnarray}
where
\begin{eqnarray}\label{eq_JDRSAM_g_func_def}
    g(x,h)
    &=&\frac{1}{2} \left(\theta+1 \right)h^T \Sigma\Sigma^T h- a_0-A_0^T x -h^T (\hat{a} + \hat{A}x)
                                                            \nonumber\\
   && +\int_{\mathbf{Z}}  \left\{\frac{1}{\theta}
        \left[\left(1+h^T \gamma(z)\right)^{-\theta}-1\right]
        +h^T \gamma(z)\mathit{1}_{\mathbf{Z}_0}(z)
        \right\} \nu(dz)
\end{eqnarray}
and the Dol\'eans exponential $\chi_t^h$ is given by~\eqref{eq_JDRSAM_Doleansexp_chi}.
\\

\subsection{Change of Measure}\label{sec_Change_of_Measure}
Let $\mathbb{P}_{h}$ be the measure on
$(\Omega,\mathcal{F}_T)$ defined via the Radon-Nikodym derivative
\begin{eqnarray}\label{eq_JDRSAM_RNder_chi}
    \frac{d\mathbb{P}_{h}}{d\mathbb{P}}
    := \chi_T^h
\end{eqnarray}
For a change of measure to be possible, we must ensure that the
following technical condition holds:
\begin{equation}
    G(z,h(s)) < 1
                                                        \nonumber
\end{equation}

This condition is satisfied iff
\begin{eqnarray}\label{cond_JDRSAM_changeofmeasure}
        h^T (s)\gamma(z)    > -1
\end{eqnarray}
a.s. $d\nu$, which was already required for $h(t)$ to be in class $\mathcal{H}$ (Condition 3 in Definition~\ref{def_JDRSAM_controlprocess_h}). Condition~\eqref{cond_JDRSAM_changeofmeasure} is endogenous to the control problem and can be interpreted as a risk management safeguard preventing the investor from investing in some of the portfolios if the jump component of these portfolios could result in the investor's bankruptcy. 
\\

Observe that $\mathbb{P}_{h}$ is a probability measure for $h \in \mathcal{A}(T)$. Then,
\begin{equation}
    W_{t}^{h} = W_t + \theta \int_{0}^{t} \Sigma^T h(s) ds
            \nonumber
\end{equation}
is a standard Brownian motion under the measure $\mathbb{P}_{h}$ and we have (recall the notation defined at \eqref{nbar})
\begin{eqnarray}
    \int_{0}^{t}\int_{\mathbf{Z}_0}\tilde{N}^{h}(ds,dz)
&=&     \int_{0}^{t}\int_{\mathbf{Z}_0}N(ds,dz)
    -   \int_{0}^{t}\int_{\mathbf{Z}_0} \left\{1-G(z,h(s))\right\}\nu(dz)ds
                \nonumber\\
&=&     \int_{0}^{t}\int_{\mathbf{Z}_0}N(ds,dz)
    -   \int_{0}^{t}\int_{\mathbf{Z}_0} \left\{\left(1+h^T \gamma(z)\right)^{-\theta}\right\} \nu(dz)ds
                \nonumber
\end{eqnarray}
As a result, $X(t)$ satisfies the SDE:
\begin{eqnarray}\label{eq_JDRSAM_diffstate_state_SDE}
    dX(t)
    &=& f\left(X(t),h(t)\right) dt
                + \Lambda dW_{t}^{h}
        , \qquad t \in [0,T]
\end{eqnarray}
where
\begin{eqnarray}\label{eq_JDRSAM_diffstate_f}
        f(x,h) := \left(b + B x - \theta\Lambda \Sigma^T h\right)
\end{eqnarray}

We will now introduce the following two auxiliary criterion functions under the
measure $\mathbb{P}_{h}$:
\begin{itemize}
\item the auxiliary function directly associated with the risk-sensitive control problem:
\begin{equation}\label{eq_JDRSAM_diffstate_auxcriterion}
    I(v,x;h;t,T) = - \frac{1}{\theta} \ln \mathbf{E}_{t,x}^{h}
        \left[ \exp \left\{ \theta \int_{t}^{T} g(X_s,h(s)) ds
        - \theta \ln{v} \right\} \right]
\end{equation}
where $\mathbf{E}_{t,x} \left[ \cdot \right]$ denotes the
expectation taken with respect to the measure $\mathbb{P}_{h}$ and with initial conditions $(t,x)$.

\item the exponentially transformed criterion
\begin{equation}\label{eq_JDRSAM_Exp_of_int_criterion}
    \tilde{I}(v,x,h;t,T)
        := \mathbf{E}_{t,x}^{h}
        \left[ \exp \left\{ \theta \int_{t}^{T} g(s,X_s,h(s))
        ds -\theta \ln v
        \right\} \right]
\end{equation}
which we will find convenient to use in our derivations.
\\
\end{itemize}

We have completed our reformulation of the problem under the measure $\mathbb{P}_{h}$.
It is striking that the asset allocation problem with jump-diffusion asset prices reduces to a stochastic control problem for a \emph{diffusion} process, with dynamics~\eqref{eq_JDRSAM_diffstate_state_SDE} and reward function~\eqref{eq_JDRSAM_diffstate_auxcriterion} or~\eqref{eq_JDRSAM_Exp_of_int_criterion}.

\subsection{The Risk-Sensitive Control Problems under $\mathbb{P}_{h}$}\label{sec_Optim}
Let $\Phi$ be the value function for the auxiliary criterion
function $I(v,x;h;t,T)$. Then $\Phi$ is defined as
\begin{equation}\label{eq_JDRSAM_diffstate_valuefunction}
    \Phi(t,x) = \sup_{h \in \mathcal{A}(T)} I(v,x;h;t,T)
\end{equation}
We will show that $\Phi$ satisfies the HJB PDE
\begin{equation}\label{eq_JDRSAM_diffstate_HJBPDE}
   \frac{\partial \Phi}{\partial t}(t,x)
   + \sup_{h \in \mathcal{J}}L_{t}^{h}\Phi(t,x) = 0
        , \qquad (t,x) \in (0,T)\times\mathbb{R}^n
\end{equation}
where
\begin{eqnarray}\label{eq_JDRSAM_diffstate_HJBPDE_operator_L}
    L_{t}^{h}\Phi(t,x)
    &=&\left(b+ B x - \theta\Lambda \Sigma^T h(s)\right)^T D\Phi
                                \nonumber\\
    &&
        + \frac{1}{2} \textrm{tr} \left( \Lambda \Lambda^T  D^2 \Phi \right)
        - \frac{\theta}{2} (D\Phi)^T \Lambda \Lambda^T  D\Phi
        - g(x,h)
\end{eqnarray}
and subject to terminal condition
\begin{equation}\label{eq_JDRSAM_diffstate_HJBPDE_termcond}
      \Phi(T, x) = \ln v,\qquad x \in \mathbb{R}^n.
\end{equation}
Similarly, let $\tilde{\Phi}$ be the value function for the auxiliary criterion
function $\tilde{I}(v,x;h;t,T)$. Then $\tilde{\Phi}$ is defined as
\begin{equation}\label{eq_JDRSAM_diffstate_exptrans_valuefunction}
    \tilde{\Phi}(t,x) = \inf_{h \in \mathcal{A}(T)} \tilde{I}(v,x;h;t,T).
\end{equation}
The corresponding HJB PDE is
\begin{equation}\label{eq_JDRSAM_diffstate_exptrans_HJBPDE}
                \frac{\partial \tilde{\Phi}}{\partial t}(t,x)
                + \frac{1}{2} \textrm{tr} \left( \Lambda \Lambda^T  D^2 \tilde{\Phi}(t,x)\right)
                        + H(t,x,\tilde{\Phi},D\tilde{\Phi})=0,  
\end{equation}
and subject to terminal condition
\begin{equation}\label{eq_JDRSAM_diffstate_exptrans_HJBPDE_termcond}
        \tilde{\Phi}(T, x) = v^{-\theta}
\end{equation}
and where
\begin{equation}\label{eq_JDRSAM_diffstate_logtrans_H_function}
   H(x,r,p) =\inf_{h \in \mathcal{J}} \left\{f(x,h)^T p + \theta g(x,h) r
        \right\}
\end{equation}
for $r \in \mathbb{R}$, $p \in \mathbb{R}^n$ and in particular,
\begin{equation}\label{eq_JDRSAM_diffstate_relationship_Phi_tildePhi}
    \tilde{\Phi}(t,x)
    = \exp \left\{-\theta \Phi(t,x) \right\}
\end{equation}

Note that since $\Phi$ and $\tilde{\Phi}$ are related through a strictly monotone continuous transformation, an admissible (optimal) strategy for the exponentially transformed problem is also admissible (optimal) for the risk-sensitive problem.

%
%

\section{Main Result}\label{sec_Mainresult}
In this section, we present the main result of this article, namely that the risk-sensitive jump diffusion problem admits a classical $(C^{1,2})$ solution, and show that the value function $\Phi$ is convex in $x$.
\\
 
\begin{proposition}\label{prop_JDRSAM_convexity_Phi}
        The value function $\Phi(t,x)$ is convex in $x$.
\end{proposition}

\begin{proof}

        To prove that the value function $\Phi(t,x)$ is convex in $x$, it is necessary and sufficient to show that $\forall (x_1,x_2) \in \mathbb{R}^n$ and for any $\kappa \in (0,1)$,
\begin{eqnarray}\label{eq_JDRSAM_Phi_convex}
        \Phi(t, \kappa x_1 + (1-\kappa) x_2)
\leq    \kappa \Phi(t, x_1)
+       (1-\kappa) \Phi(t, x_2)
\end{eqnarray}
Start from the left hand side:
\begin{eqnarray}
        &&      \Phi(t, \kappa x_1 + (1-\kappa) x_2)
                                                                                                                \nonumber\\
        &=& \sup_{h \in \mathcal{A}(T)} - \frac{1}{\theta} \ln \mathbf{E}^h_{t, \kappa x_1 + (1-\kappa) x_2}
        \left[ \exp \left\{ \theta \int_{t}^{T} g(X_s,h(s)) ds
        - \theta \ln{v} \right\} \chi(t)\right]
                                                                                                                \nonumber\\
        &=& \sup_{h \in \mathcal{A}(T)} - \frac{1}{\theta} \ln \mathbf{E}^h_{t,(x_1,x_2)}
        \left[ \exp \left\{ \theta \int_{t}^{T} g(\kappa X_1(s) + (1-\kappa) X_2(s),h(s)) ds
        - \theta \ln{v} \right\} \chi(t)\right] \nonumber\\
        &=& \sup_{h \in \mathcal{A}(T)} - \frac{1}{\theta} \ln \mathbf{E}^h_{t,(x_1,x_2)}
        \left[ \exp \left\{ \kappa \theta \int_{t}^{T} g(X_1(s),h(s)) ds
                        \right.\right. \nonumber\\
        &&      \left.\left.
                        + (1-\kappa) \theta \int_{t}^{T} g(X_2(s),h(s)) ds
                        - \theta \ln{v} \right\} \chi(t)\right] \nonumber\\
        &=& \sup_{h \in \mathcal{A}(T)} - \frac{1}{\theta} \ln \mathbf{E}^h_{t,(x_1,x_2)}
        \left[ \left(\exp \left\{ \theta \int_{t}^{T} g(X_1(s),h(s)) ds - \theta \ln{v} \right\} \chi(t) \right)^{\kappa}
                        \right.
                                                                                                                \nonumber\\
        &&      \left. \times
                        \left(\exp \left\{\theta \int_{t}^{T} g(X_2(s),h(s)) ds - \theta \ln{v} \right\} \chi(t) \right)^{1-\kappa}\right]
                                                                                                                \nonumber\\
        &\leq& \sup_{h \in \mathcal{A}(T)} - \frac{1}{\theta} \ln \left\{
                        \mathbf{E}^h_{t,x_1}
                \left[ \left(\exp \left\{ \theta \int_{t}^{T} g(X_1(s),h(s)) ds - \theta \ln{v} \right\} \chi(t) \right)^{\kappa}\right]
                        \right.
                                                                                                                \nonumber\\
        &&      \left. \times
                        \mathbf{E}^h_{t,x_2} \left[
                        \left(\exp \left\{\theta \int_{t}^{T} g(X_2(s),h(s)) ds - \theta \ln{v} \right\} \chi(t) \right)^{1-\kappa}\right]
                \right\}
                                                                                                                \nonumber
\end{eqnarray}
\begin{eqnarray}
        &=& \sup_{h \in \mathcal{A}(T)} \left\{
                - \frac{1}{\theta} \ln
                        \mathbf{E}^h_{t,x_1}
                \left[ \left(\exp \left\{ \theta \int_{t}^{T} g(X_1(s),h(s)) ds - \theta \ln{v} \right\} \chi(t) \right)^{\kappa}\right]
                \right.
                                                                                                                \nonumber\\
        &&      \left. - \frac{1}{\theta} \ln
                        \mathbf{E}^h_{t,x_2} \left[
                        \left(\exp \left\{\theta \int_{t}^{T} g(X_2(s),h(s)) ds - \theta \ln{v} \right\} \chi(t) \right)^{1-\kappa}\right]
                \right\}
                                                                                                                \nonumber\\
        &\leq& \sup_{h \in \mathcal{A}(T)} - \frac{1}{\theta} \ln
                        \mathbf{E}^h_{t,x_1}
                \left[ \left(\exp \left\{ \theta \int_{t}^{T} g(X_1(s),h(s)) ds - \theta \ln{v} \right\} \chi(t) \right)^{\kappa}\right]
                                                                                                                \nonumber\\
        &&      + \sup_{h \in \mathcal{A}(T)} - \frac{1}{\theta} \ln
                        \mathbf{E}^h_{t,x_2} \left[
                        \left(\exp \left\{\theta \int_{t}^{T} g(X_2(s),h(s)) ds - \theta \ln{v} \right\} \chi(t) \right)^{1-\kappa}\right]
                                                                                                                \nonumber\\
        &\leq& \sup_{h \in \mathcal{A}(T)} - \frac{\kappa}{\theta} \ln
                        \mathbf{E}^h_{t,x_1}
                \left[ \exp \left\{ \theta \int_{t}^{T} g(X_1(s),h(s)) ds - \theta \ln{v} \right\} \chi(t) \right]
                                                                                                                \nonumber\\
        &&      + \sup_{h \in \mathcal{A}(T)} - \frac{1-\kappa}{\theta} \ln
                        \mathbf{E}^h_{t,x_2} \left[
                        \exp \left\{\theta \int_{t}^{T} g(X_2(s),h(s)) ds - \theta \ln{v} \right\} \chi(t) \right]
                                                                                                                \nonumber\\
        &=& \kappa \Phi(t, x_1) +       (1-\kappa) \Phi(t, x_2)
                                                                                                                \nonumber
\end{eqnarray}
where the fourth equality follows from the fact that the covariance of two random variables inside the expectations is positive and the third inequality is due to the fact that the function $x \mapsto x^\alpha$ for $x > 0$ and $\alpha \in (0,1)$ is concave.
\end{proof}

\begin{corollary}\label{coro_JDRSAM_convexity_equivprop_tildePhi}
        The exponentially transformed value function $\tilde{\Phi}$ has the following property: $\forall (x_1,x_2) \in \mathbb{R}^2, \kappa \in (0,1,)$,
\begin{eqnarray}\label{eq_JDRSAM_tildePhi_property_convexityofPhi}
        \tilde{\Phi}(t, \kappa x_1 + (1-\kappa) x_2)
\geq
        \tilde{\Phi}^\kappa (t, x_1) \tilde{\Phi}^{1-\kappa} (t, x_2)
\end{eqnarray}
\end{corollary}

\begin{proof}
The properties follows immediately from the definition of $\Phi= -\frac{1}{\theta} \ln \tilde{\Phi}$.
\end{proof}
\\

\noindent We now come to the main result of this article. Recall the standing assumptions: in Section \ref{as}, Assumption~\ref{as_assetjumps_upanddown_1} is a condition on the support of the jump measure, Assumptions~\ref{as_JDRSAM_sigposdef} and ~\ref{as_JDRSAM_lamposdef} are the non-degeneracy conditions $\Sigma\Sigma^T>0,\,\Lambda\Lambda^T>0$, while Assumption~\ref{as_JDRSAM_A_rank_n}, introduced in Section \ref{zb} below, is a full-rank condition on the matrix $\hat{A}$ defined at \eqref{eq_JDRSAM_V_dynamics}.
\\

\begin{theorem}\label{theo_JDRSAM_diffstate_main_result}
Under Assumptions~\ref{as_assetjumps_upanddown_1}--~\ref{as_JDRSAM_A_rank_n}
the following hold:
\begin{enumerate}[1.]
\item the optimal asset allocation is the unique maximizer of the supremum~\eqref{eq_JDRSAM_diffstate_HJBPDE}

\item $\tilde{\Phi}$ is the unique $C^{1,2}\left([0,T]\times \mathbb{R}^n\right)$ solution  of the RS HJB PIDE~\eqref{eq_JDRSAM_diffstate_exptrans_HJBPDE}-\eqref{eq_JDRSAM_diffstate_exptrans_HJBPDE_termcond}. Moreover, $\tilde{\Phi}$ satisfies the property~\eqref{eq_JDRSAM_tildePhi_property_convexityofPhi}

\item $\Phi$ is the unique $C^{1,2}\left([0,T]\times \mathbb{R}^n\right)$ solution  of the RS HJB PIDE~\eqref{eq_JDRSAM_diffstate_HJBPDE}-\eqref{eq_JDRSAM_diffstate_HJBPDE_termcond}. Moreover, $\Phi$ is convex in its argument $x$.
\\
\end{enumerate}

\end{theorem}

\begin{proof} The proof is based on a series of results proved in Sections \ref{sec-existence}--\ref{sec-classical} below. These combine to give us the following argument.\\

\emph{Existence of an optimal control - } by Proposition~\ref{prop_JDRSAM_diffstate_optimcontrol}, the supremum in~\eqref{eq_JDRSAM_diffstate_HJBPDE} admits a unique Borel measurable maximizer. Moreover, by Proposition~\ref{prop_JDRSAM_diffstate_hstar_admissible}, this maximizer is admissible and by Proposition~\ref{prop_JDRSAM_diffstate_OptimalJandI} it is also a maximizer with respect to the $\mathbb{P}$-measure criterion $J$ defined in~\eqref{eq_JDRSAM_criterion_J}. Thus, we can take this maximizer as our optimal asset allocation.
\\

\emph{Existence of a classical ($C^{1,2}$) solution - } by Corollary~\ref{Coro_JDRSAM_diffstate_existence}, $\tilde{\Phi}$ is a $C^{1,2}\left([0,T]\times \mathbb{R}^n\right)$ solution  of the RS HJB PDE~\eqref{eq_JDRSAM_diffstate_exptrans_HJBPDE}-\eqref{eq_JDRSAM_diffstate_exptrans_HJBPDE_termcond}.
\\

\emph{Uniqueness of the classical solution - }   the existence of zero beta policies enable us to deduce (as in Step 1 of the proof of Theorem~\ref{Theo_JDRSAM_diffstate_existence}) that $\tilde{\Phi}$ is bounded. Part (i). of Verification Theorem~\ref{Theo_JDRSAM_verification} therefore applies. Choosing as optimal control the unique maximizer of the supremum in~\eqref{eq_JDRSAM_diffstate_HJBPDE}, part (ii). of Theorem~\ref{Theo_JDRSAM_verification} also applies: $\tilde{\Phi}$ is the unique solution to the HJB PIDE. Property~\eqref{eq_JDRSAM_tildePhi_property_convexityofPhi} is proved in Corollary~\ref{coro_JDRSAM_convexity_equivprop_tildePhi}.
\\

By Corollaries~\ref{Coro_JDRSAM_diffstate_existence} and~\ref{Coro_JDRSAM_verification}, it then follows that $\Phi$ is the unique classical solution to the HJB PIDE~\eqref{eq_JDRSAM_diffstate_HJBPDE} with terminal condition \eqref{eq_JDRSAM_diffstate_HJBPDE_termcond}. Moreover, $\Phi$ is convex in its second argument $x$.
\\

\end{proof}

This result proves that we have solved our original control problem in the context of strong, classical solutions. What would this imply in terms of weaker viscosity solutions? As a classical solution is also a viscosity solution, our result implies that the value function is indeed a viscosity solution of the HJB PDE. However, uniqueness of classical solutions does not necessarily imply uniqueness of viscosity solutions. To prove uniqueness in the viscosity sense, we would need a comparison result such as Theorem 33 in Davis and Lleo~\cite{dall_JDRSAM_Visc}.
\\

In the remainder of this paper, we develop the various technical arguments required in the proof of Theorem~\ref{theo_JDRSAM_diffstate_main_result}.
\\

%
%

\section{Existence of a Maximizing Control\\}\label{sec-existence}
\begin{proposition}\label{prop_JDRSAM_diffstate_optimcontrol}
        Under Assumption~\ref{as_JDRSAM_sigposdef}, the supremum in~\eqref{eq_JDRSAM_diffstate_HJBPDE} admits a unique Borel measurable maximizer $\hat{h}(t,x,p)$ for $(t,x,p) \in [0,T]\times\mathbb{R}^n\times\mathbb{R}^n$. Moreover, the maximizer $\hat{h}(t,x,p)$ is an interior point of the set $\overline{\mathcal{J}}$. 
\\
\end{proposition}

\begin{proof}
The supremum in~\eqref{eq_JDRSAM_diffstate_HJBPDE} can be expressed as
\begin{eqnarray}\label{eq_JDRSAM_diffstate_supL_deriv}
        && \sup_{h \in \mathcal{J}} L_{t}^{h}\Phi
                    \nonumber\\
        &=&
            \left( b+ Bx \right)^T D\Phi
            + \frac{1}{2} \textrm{tr} \left( \Lambda \Lambda^T  D^2 \Phi \right)
            - \frac{\theta}{2} (D\Phi)^T \Lambda \Lambda^T  D\Phi
            +a_0+A_0^T x
            \nonumber\\
        &&
            +\sup_{h \in \mathcal{J}} \left\{
            - \frac{1}{2} \left(\theta+1 \right)h^T \Sigma\Sigma^T h
            -\theta h^T \Sigma\Lambda^T D\Phi
            +h^T (\hat{a} + \hat{A}x)
                \right. \nonumber\\
            &&\left.
            -\frac{1}{\theta}\int_{\mathbf{Z}}\left\{
                                                \left[\left(1+h^T \gamma(z)\right)^{-\theta}-1\right]
                                        +\theta h^T \gamma(z)\mathit{1}_{\mathbf{Z}_0}(z)
                    \right\}\nu(dz)
            \right\}
\end{eqnarray}
Under Assumption~\ref{as_JDRSAM_sigposdef}, for any $p \in \mathbb{R}^n$ the terms
\begin{eqnarray}
    - \frac{1}{2} \left(\theta+1 \right)h^T \Sigma\Sigma^T h
    -\theta h^T \Sigma\Lambda^T p
    +h^T (\hat{a} + \hat{A}x)
    - \int_{\mathbf{Z}} h^T \gamma(z)\mathit{1}_{\mathbf{Z}_0}(z)\nu(dz)
                                        \nonumber
\end{eqnarray}
and
\begin{eqnarray}
        -\frac{1}{\theta}\int_{\mathbf{Z}}
       \left\{
               \left[\left(1+h^T \gamma(z)\right)^{-\theta}-1\right]
           \right\}\nu(dz)
                                                                                                                                \nonumber
\end{eqnarray}
are both strictly concave in $h$ $\forall z \in \mathbb{Z}$ a.s. $d\nu$. Therefore, the supremum is reached for a unique maximizer $\hat{h}(t,x,p)$, which is an interior point of the set $\overline{\mathcal{J}}$, which is the closure of the set $\mathcal{J}$ defined in equation~\eqref{def_JDRSAM_setJ}, and the supremum, evaluated at $\hat{h}(t,x,p) \in \mathbb{R}^n$, is finite. By measurable selection, $\hat{h}$ can be taken as a Borel measurable function on $[0,T]\times\mathbb{R}^n\times\mathbb{R}^n$.
\\
\end{proof}

%
%

\section{Verification Theorems}\label{sec-verification}
In this section, we prove a verification theorem to the effect that if~\eqref{eq_JDRSAM_diffstate_HJBPDE} has a $C^{1,2}$ solution then that solution is equal to $\Phi$ defined by~\eqref{eq_JDRSAM_diffstate_valuefunction} and the control $h^*(t) = \hat{h}(t,x,D\Phi)$ is optimal. We will first prove a verification theorem for the exponentially transformed problem~\eqref{eq_JDRSAM_diffstate_exptrans_valuefunction} with HJB PDE~\eqref{eq_JDRSAM_diffstate_exptrans_HJBPDE} and value function $\tilde{\Phi}(t,x)$. As a corollary, we will obtain a verification theorem for the risk sensitive control problem with~\eqref{eq_JDRSAM_diffstate_valuefunction}, HJB PDE~\eqref{eq_JDRSAM_diffstate_HJBPDE} and value function $\Phi(t,x)$. Define the first order operator
\begin{eqnarray}\label{eq_JDRSAM_logtrans_HJBPDE_operator_L}
    \tilde{L}^{h}\varphi(t,x)
    &=& \left(b + B x - \theta \Lambda \Sigma^T h
         \right)^T D\varphi(t,x)
                + \theta g(x,h) \varphi(t,x)
\end{eqnarray}
\\

\begin{theorem}[Verification Theorem for the Exponentially Transformed Control Problem]\label{Theo_JDRSAM_verification}
        Let $\tilde{\phi}$ be a $C^{1,2}\left([0,T]\times \mathbb{R}^n\right)$ bounded function.
\begin{enumerate}[(i)]
        \item Assume that $\tilde{\phi}(T,x) \leq e^{-\theta \ln v} \; \forall x \in \mathbb{R}^n$ and
\begin{eqnarray}
   &&  \frac{\partial \tilde{\phi}}{\partial t}(t,x)
                + \frac{1}{2} \textrm{tr} \left( \Lambda \Lambda^T  D^2 \tilde{\phi}(t,x)\right)
                        + H(t,x,\tilde{\phi},D\tilde{\phi})
    \geq        0
                                                                                                                \nonumber
\end{eqnarray}
        on $[0,T]\times \mathbb{R}^n$, then $\tilde{\phi}(t,x) \leq \tilde{\Phi}(t,x) \; \forall (t,x) \in [0,T]\times\mathbb{R}^n$

        \item Further assume that $\tilde{\phi}(T,x) = e^{-\theta \ln v} \; \forall x \in \mathbb{R}^n$ and there exists a Borel-measurable minimizer $\tilde{h}^*(t,x)$ of $\tilde{h} \mapsto \tilde{L}^{\tilde{h}}\tilde{\phi}$ defined in~\eqref{eq_JDRSAM_logtrans_HJBPDE_operator_L} such that
\begin{eqnarray}
   &&  \frac{\partial \tilde{\phi}}{\partial t}(t,x)
                + \frac{1}{2} \textrm{tr} \left( \Lambda \Lambda^T  D^2 \tilde{\phi}(t,x)\right)
                        + H(t,x,\tilde{\phi},D\tilde{\phi})
                                                                                                        \nonumber\\
        &=&     \frac{\partial \tilde{\phi}}{\partial t}(t,x)
                + \frac{1}{2} \textrm{tr} \left( \Lambda \Lambda^T  D^2 \tilde{\phi}(t,x)\right)
                        + \tilde{L}^{\tilde{h}^*}\tilde{\phi}
                                                                                                        \nonumber\\
    &=& 0
                                                                                                                \nonumber
\end{eqnarray}
                and the stochastic differential equation
\begin{eqnarray}
    dX(t)
    &=&     \left(b + B X(t) - \theta\Lambda \Sigma^T h(t)\right) dt
                + \Lambda dW_{t}^{\theta}
                                                                                                                \nonumber
\end{eqnarray}
defines a unique solution $X(s)$ for each given initial data $X_t = x$ and the process $\pi^*(s) := \tilde{h}^*(s, X(s))$ is a well-defined control process in $\tilde{\mathcal{A}}(T)$. Then $\tilde{\phi} = \tilde{\Phi}$ and $\pi^*(s)$ is an optimal Markov control process.
\end{enumerate}

\end{theorem}

\begin{proof}
The following proof is based on an argument used by Touzi~\cite{to02}.

\begin{enumerate}[(i).]
\item Let $\tilde{h} \in \tilde{\mathcal{A}}(T)$ be an arbitrary control, with $X(t)$ the state process with initial data $X(t) = x$. Define the stopping time
\begin{equation}
        \tau_N := T \wedge \inf \left\{ s > t : |X_s - x| \geq N \right\}
                                                                                                                                \nonumber
\end{equation}

Define $Z(s) = \theta \int_{t}^{s}
g(s,X_s,\hat{h}_s) ds$, then
\begin{equation}
    d\left(e^{Z_s}\right) := \theta g(s,X_s,\hat{h}_s)e^{Z_s}   \nonumber
\end{equation}
Also, by It\^o, for $s \in \left[t, \tau_{\delta} \right]$,
\begin{equation}
   d\tilde{\phi}_s
   = \left\{
        \frac{\partial \tilde{\phi}}{\partial s}
   +   \mathcal{L}\tilde{\phi} \right\} ds
        +    D\tilde{\phi}^T \Lambda dW_{s}^{\theta}
                                                                \nonumber
\end{equation}
where $\mathcal{L}$ is the generator of the state process $X(t)$ defined as:
\begin{eqnarray}
    \mathcal{L}\tilde{\phi}(t,x)
    &:=&
        \left( b + B x - \theta \Lambda \Sigma^T h(s) \right)^T D\tilde{\phi}
    +       \frac{1}{2} \textrm{tr} \left( \Lambda \Lambda^T D^2 \tilde{\phi} \right)
                                                    \nonumber
\end{eqnarray}
\\

By the It\^o product rule, and since $dZ_s \cdot \tilde{\phi}_s = 0$, we get
\begin{equation}
        d\left(\tilde{\phi}_s e^{Z_s}\right)
    =   \tilde{\phi}_s d\left(e^{Z_s}\right) + e^{Z_s} d\tilde{\phi}_s
                                        \nonumber
\end{equation}
and hence for $s \in [t, \tau_{N}]$
\begin{eqnarray}
    \tilde{\phi}(s,X_s)e^{Z_s} &=&
        \tilde{\phi}(t,x)e^{Z_{t}}
        + \theta\int_{t}^{s}
                \tilde{\phi}(u,X_u)g(u,X_u,\hat{h}_u)e^{Z_u}du
                                                    \nonumber\\
    &&  + \int_{t}^{s}
            \left( \frac{\partial \tilde{\phi}}{\partial u}(u,X_u)
            + \mathcal{L}\tilde{\phi}(u,X_u) e^{Z_u}\right)du
        + \int_{t}^{s} D\tilde{\phi}^T \Lambda dW_{u}^{\theta}
                                                    \nonumber
\end{eqnarray}

Because for an arbitrary control $h$,
\begin{eqnarray}
        &&              \frac{\partial \tilde{\phi}}{\partial t}(t,x)
                + \frac{1}{2} \textrm{tr} \left( \Lambda \Lambda^T  D^2 \tilde{\phi}(t,x)\right)
      + \mathcal{L}^{\tilde{h}}\tilde{\phi}(t,X_t)
                                                                                                        \nonumber\\
        &\geq&
        \frac{\partial \tilde{\phi}}{\partial t}(t,x)
                + \frac{1}{2} \textrm{tr} \left( \Lambda \Lambda^T  D^2 \tilde{\phi}(t,x)\right)
                        + H(t,x,\tilde{\phi},D\tilde{\phi})
                                                                                                        \nonumber\\
   &\geq&       0
                                                                                                        \nonumber
\end{eqnarray}
and $e^{Z_s} \geq 0 \; \forall s \in [t,\tau_N]$ we have
\begin{eqnarray}
    \tilde{\phi}(t,x)e^{Z_{t}} &\leq&
       \tilde{\phi}(s,X_s)e^{Z_s}
       + \int_{t}^{s} D\tilde{\phi}^T \Lambda dW_{u}^{\theta}
                                                                                                                                \nonumber
\end{eqnarray}
Taking the expectation, we obtain
\begin{eqnarray}
    \tilde{\phi}(t,x)e^{Z_{t}}
        &\leq&
      \mathbf{E}_{t,x}^{\tilde{h},\theta} \left[ \tilde{\phi}(s,X_s)e^{Z_s} \right]
        =       \mathbf{E}_{t,x}^{\tilde{h},\theta} \left[ \tilde{\phi}(s,X_s)e^{\theta \int_{t}^{s}
g(u,X_u,\hat{h}_u) du} \right]
                                                    \nonumber
\end{eqnarray}
\\

In particular, take $s = \tau_N$ and note that $e^{Z_{t}} = 1$, then
\begin{eqnarray}
    \tilde{\phi}(t,x)e^{Z_{t}}
        &\leq&
                \mathbf{E}_{t,x}^{\tilde{h},\theta} \left[ \tilde{\phi}(\tau_N,X_{\tau_N})e^{\theta \int_{t}^{\tau_{N}}
g(u,X_u,\hat{h}_u) du} \right]
                                                    \nonumber
\end{eqnarray}
\\

Since $\tilde{\phi}$ is assumed to be bounded, there exists a constant $C_1 > 0$ such that:
\begin{eqnarray}
        \Big|
                \tilde{\phi}(s,X_s)
                e^{\theta \int_{s}^{\tau_{N}}g(s,X_s,\hat{h}_s) ds}
        \Big|
        \leq C_1e^{\theta \int_{t}^{\tau_{N}} g(u,X_u,\hat{h}_u) du}
                                                    \nonumber
\end{eqnarray}
\\

Since for an arbitrary admissible control $\tilde{h} \in \mathcal{A}(T)$ and fixed $s \in [t,T]$ there exists some constant $C_2 >0$ such that
\begin{equation}
        \left| g(s,X_s,\hat{h}_s) \right| \leq C_2 \left| 1 + X(s) \right|
                                                                                                                                        \nonumber
\end{equation}
Then
\begin{eqnarray}
        \Big|
                \tilde{\phi}(s,X_s)
                e^{\theta \int_{s}^{\tau_{N}}g(s,X_s,\hat{h}_s) ds}
        \Big|
        &\leq& C_3e^{\theta \int_{t}^{\tau_{N}} \left| 1 + X(s) \right| du}
                                                                                                                                        \nonumber\\
        &\leq& C_3e^{\theta(\tau_N -t) + \theta \int_{t}^{\tau_{N}} \left| X(s) \right| du}
                                                                                                                                        \nonumber\\
        &\leq& C_4e^{\theta \int_{t}^{\tau_{N}} \left| X(s) \right| du}
                                                                                                                                        \nonumber\\
        &\leq& C_4e^{\theta (T-t) \sup_{t \leq s \leq T}\left| X(s) \right|}
                                                                                                                                        \nonumber
\end{eqnarray}
for $C_3 = C_1 e^{C_2}$ and $C_4 = C_3 e^{\theta(T -t)}$.
\\

By the dominated convergence theorem and the assumption that $\tilde{\phi}(T,X_t) \leq e^{-\theta \ln v}$,
\begin{eqnarray}
        \tilde{\phi}(t,x)
        &\leq& \mathbf{E}_{t,x}^{\tilde{h},\theta} \left[ \tilde{\phi}(T,X_{T})e^{\theta \int_{t}^{T}
g(u,X_u,\hat{h}_u) du} \right]
                                                                                                                        \nonumber\\
        &\leq& \mathbf{E}_{t,x}^{\tilde{h},\theta} \left[e^{\theta \int_{t}^{T}
g(u,X_u,\hat{h}_u) du} - \theta \ln v\right]
                                                                                                                        \nonumber
\end{eqnarray}
We have now proved the first part of the theorem.

\item To prove the second part, we can simply apply the same reasoning for the optimal control $\tilde{h}^*$. Note, however, that with this choice of control we would have
\begin{eqnarray}
   &&  \frac{\partial \tilde{\phi}}{\partial t}(t,x)
                + \frac{1}{2} \textrm{tr} \left( \Lambda \Lambda^T  D^2 \tilde{\phi}(t,x)\right)
                        + H(t,x,\tilde{\phi},D\tilde{\phi})
                                                                                                        \nonumber\\
        &=&     \frac{\partial \tilde{\phi}}{\partial t}(t,x)
                + \frac{1}{2} \textrm{tr} \left( \Lambda \Lambda^T  D^2 \tilde{\phi}(t,x)\right)
                        + \tilde{L}^{\tilde{h}^*}\tilde{\phi}
                                                                                                        \nonumber\\
    &=& 0
                                                                                                        \nonumber
\end{eqnarray}
which would lead us to equality in the last equation, i.e.
\begin{eqnarray}
        \tilde{\phi}(t,x)
        &=& \mathbf{E}_{t,x}^{\tilde{h},\theta} \left[e^{\theta \int_{t}^{T}
g(u,X_u,\hat{h}_u) du} - \theta \ln v\right]
                                                                                                                        \nonumber
\end{eqnarray}
\end{enumerate}
\end{proof}

\begin{corollary}[Verification Theorem for the Risk-Sensitive Control Problem]\label{Coro_JDRSAM_verification}
        Let $\phi$ be a $C^{1,2}\left([0,T]\times \mathbb{R}^n\right) \cap C\left([0,T]\times \mathbb{R}^n\right)$ bounded function.
\begin{enumerate}[(i)]
        \item Assume that $\phi(T,x) \leq e^{-\theta \ln v} \; \forall x \in \mathbb{R}^n$ and
\begin{eqnarray}
    \frac{\partial \phi}{\partial t}
    + \sup_{h \in \mathcal{J}}
    L_{t}^{h}\phi(t,X(t))
    \geq        0
                                                                                                                \nonumber
\end{eqnarray}
        on $[0,T]\times \mathbb{R}^n$, then $\phi(t,x) \leq \tilde{\Phi}(t,x) \; \forall (t,x) \in [0,T]\times\mathbb{R}^n$

        \item Further assume that $\phi(T,x) = e^{-\theta \ln v} \; \forall x \in \mathbb{R}^n$ and there exists a minimizer $h^*(t,x)$ of $h \mapsto L^{h}\phi$ defined in~\eqref{eq_JDRSAM_diffstate_HJBPDE_operator_L} such that
\begin{eqnarray}
        \frac{\partial \phi}{\partial t}
    + \sup_{h \in \mathcal{J}}
    L_{t}^{h}\phi(t,X(t))
        =
        \frac{\partial \phi}{\partial t}
    + L_{t}^{h^*}\phi(t,X(t))
    =   0
                                                                                                                \nonumber
\end{eqnarray}
                and the stochastic differential equation
\begin{eqnarray}
    dX(s)
    &=&     \left(b + B X(s^-) - \theta\Lambda \Sigma^T h(s)\right) ds
                                + \Lambda dW_{s}^{\theta}
                                                                                                                \nonumber
\end{eqnarray}
defines a unique solution $X$ for each given initial data $X_t = x$ and the process $\pi^*(s) := \tilde{h}^*(s, X(s))$ is a well-defined control process in $\tilde{\mathcal{A}}(T)$. Then $\phi = \Phi$ and $\pi^*(s)$ is an optimal Markov control process.
\end{enumerate}
\end{corollary}

\begin{proof}
        This corollary follows from equation~\eqref{eq_JDRSAM_diffstate_relationship_Phi_tildePhi} and from the fact that an admissible (optimal) strategy for the exponentially transformed problem is also admissible (optimal) for the risk-sensitive problem. 
\end{proof}
\\

\begin{proposition}\label{prop_JDRSAM_diffstate_hstar_admissible}
The process $h^*(t)$ is admissible: $h^*(t) \in \mathcal{A}(T)$.
\end{proposition}

\begin{proof}
        Refer to Appendix~\ref{sec_Admissibility} for a full discussion and a proof of this proposition.
\end{proof}
\\

Applying Proposition~\ref{prop_JDRSAM_diffstate_hstar_admissible} we deduce that the control $h^*(t)$ is optimal for the auxiliary problems~\eqref{eq_JDRSAM_diffstate_auxcriterion} and~\eqref{eq_JDRSAM_Exp_of_int_criterion} resulting from the change of measure. However, this proposition is not sufficient to conclude that  $h^*(t)$ is optimal for the original problem~\eqref{eq_JDRSAM_criterion_J} set under the $\mathbb{P}$-measure. The next result show that this is indeed the case.
\\

\begin{proposition}\label{prop_JDRSAM_diffstate_OptimalJandI}
        The optimal control $h^*(t)$ for the auxiliary problem $$\sup_{h \in \mathcal{A}(T)}I(v,x;h;t,T)$$ where $I$ is defined in~\eqref{eq_JDRSAM_diffstate_auxcriterion} is also optimal for the initial problem  $\sup_{h \in \mathcal{A}(T)}J(x,t,h)$ where $J$ is defined in~\eqref{eq_JDRSAM_criterion_J}.
\end{proposition}

\begin{proof}
        See Appendix~\ref{sec_Admissibility}.
\end{proof}
\\

%
%

\section{Existence of a Classical Solution}\label{sec-classical}
Historically, proving the existence of a strong, analytical solution to the HJB PDE was both the main difficulty and the main objective when solving a control problem. Fleming and Rishel~\cite{fl75} as well as Krylov~\cite{kr80} and~\cite{kr87} have been the main contributors, proposing techniques based either on PDE arguments or on probability theory. Recently however, the emphasis has switched from strong solutions to weaker types of solution. Viscosity solutions have proved particularly useful and successful, gaining many applications in stochastic control theory (see for example the classic article by Crandall, Ishii and Lions~\cite{crisli92} as well as Fleming and Soner~\cite{flso06} for their applications to stochastic control). The reason for this appeal is twofold. First, it is significantly easier to prove the existence of a viscosity solution than a classical solution. In the viscosity world, the difficulty is shifted from proof of existence to proof of uniqueness, and even then it is generally easier to prove uniqueness of a viscosity solution via a comparison theorem than the existence of a classical solution. Second, the stability result due to Barles and Souganidis~\cite{baso91} connects directly viscosity solutions to numerical methods, making it easy to solve `real world' control problems.
\\

This section follows similar arguments to those developed by Fleming and Rishel~\cite{fl75} (Theorem 6.2 and Appendix E). Namely, we use an approximation in policy space alongside results on linear parabolic partial differential equations to prove that the exponentially transformed value functions $\tilde{\Phi}$ is of class $C^{1,2}((0,T)\times \mathbb{R}^n)$. Then it follows that the value functions $\Phi$ is also of class $C^{1,2}((0,T)\times \mathbb{R}^n)$. The approximation in policy space algorithm was originally proposed by Bellman in the 1950s (see Bellman~\cite{be57} for details) as a numerical method to compute the value function. Our approach has two steps. First, we use the approximation in policy space algorithm to show existence of a classical solution in a bounded region. Then, we extend our argument to unbounded state space. To derive this second result we follow a different argument than Fleming and Rishel~\cite{fl75} which makes more use of the actual structure of the control problem.
\\

Our interest in classical solutions is as much mathematical as practical. First, since a smooth solution is a viscosity solution but the converse is not necessarily true, we are proving a stronger result. Second, this stronger result immediately translates a better grasp of the analytical properties of the value function. While viscosity solutions provide continuity, they do not generally give information about higher order derivatives. By contrast, classical solutions are smooth in the state, implying that they are (at least) $C^1$ in time and $C^2$ in the state. Third, viscosity solutions are purely about solving the PDE and although they show that the value function is the unique solution of the HJB PDE they do not prove directly the control problem has a solution, that is a pair of a value function and an admissible optimal control. Fourth and finally, in our case seeking a strong solution does not impair our search for numerical results. Because our state process $X(t)$ can clearly be interpreted as the continuous time limit of a Markov Chain, we can apply well-known results by Kushner and Dupuis~\cite{kudu01} to prove convergence of a finite approximation scheme to the value function. We can therefore solve concrete portfolio selection problems quite directly.

\subsection{``Zero Beta'' Policies}\label{zb}
In this section, we introduce a new class of control policies: the ``zero beta'' ($0\beta$) policies:
\\

\begin{definition}[$0\beta$-policy]\label{def_JDRSAM_diffstate_ZeroBetaPolicy}
        By reference to the definition of the function $g$ in equation~\eqref{eq_JDRSAM_g_func_def}, a \emph{`zero beta' ($0\beta$) control policy} $\check{h}(t)$ is an admissible control policy for which the function $g$ is independent of the state variable $x$.
\end{definition}
\\

The term `zero beta' is borrowed from financial economics (see for instance Black~\cite{bl72}). To avoid assuming the existence of a globally risk-free rate in factor models such as the CAPM, the APT or in ad-hoc valuation models, it is customary to build portfolios without any exposure to the factor(s) as a substitute for the risk-free rate. These special portfolios are referred to as `zero beta' portfolios by reference to the slope coefficient $\beta$ used to measure the sensitivity of asset returns to the valuation factor(s).
\\

In the risk sensitive asset management model, if $A_0 = 0$, then the policy $h^{0} = 0$, i.e. invest all the wealth in the risk-free asset, is a $0\beta$-policy. When $A_0 \neq 0$, we see from~\eqref{eq_JDRSAM_V_dynamics} that a $0\beta$-policy can exist only if there is a vector $\check{h}$ satisfying
\begin{equation}       \check{h}^T \hat{A} = - A_0.\label{check}\end{equation}
We introduce the following standing assumption.
\begin{assumption}\label{as_JDRSAM_A_rank_n}
        The matrix $\hat{A}$ has rank $n$.
\end{assumption}

Under this assumption there are always $0\beta$-policies: we have only to take a vector $\check{h}$ satisfying \eqref{check} and scale it if necessary so that $\check{h}\in\mathcal{J}$ (see \eqref{def_JDRSAM_setJ}) and then $h(t,\omega)=\check{h}$ is a $0\beta$-policy. We have no reason to consider anything other than the set $\mathcal{Z}$ of constant policies of this kind. Note that when $\check{h}\in{\mathcal Z}$ the function $g$ of \eqref{eq_JDRSAM_g_func_def} is a constant, $g(x,\check{h})\equiv\check{g}$.

\subsection{The $L^{\eta}(K)$ and $\mathscr{L}^{\eta}(K), 1 < \eta < \infty$ Spaces}
The following ideas and notations relate to the treatment of linear parabolic partial differential equations found in Ladyzhenskaya, Solonnikov and Uralceva~\cite{la68}. The relevant results are summarized in Appendix E of Fleming and Rishel. They concern PDEs of the form
\begin{eqnarray}\label{eq_JDRSAM_diffstate_generic_linear_PDE}
        \frac{\partial \psi}{\partial t}
        + \frac{1}{2} \textrm{tr}
                    \left( a(t,x) D^2 \psi\right)
                +       b(t,x)^T  D\psi
        + \theta c(t,x) \psi
        + d(t,x)
    =   0
\end{eqnarray}
on a set $Q = (0,T) \times G$ and with boundary condition
\begin{eqnarray}
    \psi(t,x) &=& \Psi_T(x)    \quad x \in G
                            \nonumber\\
    \psi(t,x) &=& \Psi(t,x)    \quad (t,x) \in (0,T) \times \partial G
                            \nonumber
\end{eqnarray}
The set $G$ is open and is such that $\partial G$ is a compact manifold of class $C^2$. Denote by
\begin{itemize}
        \item $\partial^*Q$ the boundary of Q, i.e.
                \begin{equation}
                \partial^*Q := \left( \left\{ T \right\} \times G \right)
                                                                \cup \left( (0,T) \times \partial G \right)
                            \nonumber
                \end{equation}
        \item $L^{\eta}(K)$ the space of $\eta$-th power integrable functions on $K \subset Q$;
        \item $\lVert \cdot \rVert_{\eta,K}$ the norm in $L^{\eta}(K)$.
\end{itemize}
Also, denote by $\mathscr{L}^{\eta}(Q), 1 < \eta < \infty$ the space of all functions $\psi$ such that $\psi$ and all its generalized partial derivatives are in $L^{\eta}(K)$. We associate with this space the Sobolev-type norm:
\begin{eqnarray}
        \lVert \psi \rVert_{\eta,K}^{(2)}
        :=      \lVert \psi \rVert_{\eta,K}
        +       \Big\lVert \frac{\partial \psi}{\partial t} \Big\rVert_{\eta,K}
        +       \sum_{i=1}^n \Big\lVert \frac{\partial \psi}{\partial x_i} \Big\rVert_{\eta,K}
        +       \sum_{i,j=1}^n \Big\lVert \frac{\partial^2 \psi}{\partial x_i x_j} \Big\rVert_{\eta,K}
\end{eqnarray}

We will also introduce additional notation and concepts as required in the proofs.
\\

\subsection{Existence of a Classical Solution}
In this section, we use an approximation in policy space to show the existence of a $C^{1,2}$ solution to the RS HJB PDE~\eqref{eq_JDRSAM_diffstate_HJBPDE}.
\\

%
%
%
%

\begin{theorem}[Existence of a Classical Solution for the Exponentially Transformed Control Problem]\label{Theo_JDRSAM_diffstate_existence}
        The RS HJB PDE~\eqref{eq_JDRSAM_diffstate_exptrans_HJBPDE} with terminal condition $\tilde{\Phi}(T,x) = e^{-\theta\ln v}$ has a solution $\tilde{\Phi} \in C^{1,2}\left((0,T)\times\mathbb{R}^n\right)$ with $\tilde{\Phi}$ continuous in $[0,T]\times\mathbb{R}^n$.
\end{theorem}
\\

\begin{proof}
\textbf{Step 1: Approximation in policy space - bounded space}\\
Consider the following auxiliary problem: fix $R>0$ and let $\mathscr{B}_R$ be the open $n$-dimensional ball of radius $R>0$ centered at 0 defined as $\mathscr{B}_R := \left\{x \in \mathbb{R}^n: |x|<R \right\}$. We construct an investment portfolio by solving the optimal risk-sensitive asset allocation problem as long as $X(t) \in \mathscr{B}_R$ for $R > 0$. Then, as soon as $X(t) \notin \mathscr{B}_R$, we switch all of the wealth into the $0\beta$ policy $\check{h}$ from the exit time $t$ until the end of the investment horizon at time $T$. The HJB PDE for this auxiliary problem can be expressed as
\begin{eqnarray}\label{as_jDRSAM_diffstate_auxiliaryPDE_tildePhi}
        \frac{\partial \tilde{\Phi}}{\partial t}
        + \frac{1}{2} \textrm{tr}
                    \left( \Lambda \Lambda^T (t) D^2 \tilde{\Phi}\right)
                +       H(t,x,\tilde{\Phi},D\tilde{\Phi})
    =   0
        \qquad \forall (t,x) \in Q_R := (0,T) \times \mathscr{B}_R
                                                                                                                \nonumber\\
\end{eqnarray}
where, as in~\eqref{eq_JDRSAM_diffstate_logtrans_H_function}
\begin{eqnarray}
   H(x,r,p) &=&
                \inf_{h \in \overline{\mathcal{J}}} \left\{
        f(x,h)^T  p + \theta g(t,x,h) r
    \right\}
                                                                                                \nonumber
\end{eqnarray}
for $p \in \mathbb{R}^n$ and for $g$ and $f$ defined in~\eqref{eq_JDRSAM_g_func_def} and~\eqref{eq_JDRSAM_diffstate_f} respectively. $\overline{\mathcal{J}}$, defined as the closure of $\mathcal{J}$ is a compact set. The boundary conditions are
\begin{eqnarray}
        \tilde{\Phi}(t, x) &=& \Psi(t,x)
        \qquad \forall (t,x) \in \partial^* Q_R := \left( (0,T) \times \partial \mathscr{B}_R \right) \cup \left( \left\{T\right\} \times \mathscr{B}_R \right)
                                                \nonumber
\end{eqnarray}
with
\begin{itemize}
\item $\Psi(T, x) = e^{-\theta\ln v} \; \forall x \in  \mathscr{B}_R$;
\item $\Psi(t, x) :=  \psi(t,x) := e^{\theta \check{g} (T-t)} \; \forall (t,x) \in (0,T) \times \partial \mathscr{B}_R$ and where $\check{h}$ is a fixed arbitrary $0\beta$ policy which is constant as a function of time. Note that $\psi$ is obviously of class $C^{1,2}(\overline{Q_R})$ and that the Sobolev-type norm
\begin{eqnarray}\label{as_JDRSAM_diffstate_assumption_Psi}
        \lVert \Psi \rVert_{\eta,\partial^*Q_R}^{(2)}
        =       \lVert \tilde{\Psi} \rVert_{\eta,Q_R}^{(2)}
\end{eqnarray}
is finite.
\\
\end{itemize}

Define a sequence of functions $\tilde{\Phi}^1$, $\tilde{\Phi}^2$,... $\tilde{\Phi}^k$,... on $\overline{Q_R}=[0,T]\times\overline{\mathscr{B}_R}$ and of bounded measurable feedback control laws $h^{0}$, $h^{1}$,... $h^{k}$,... where $h^{0}$ is an arbitrary control. Assuming $h^{k-1}$ is known, we define the function $\tilde{\Phi}^{k+1}$ as the solution to the boundary value problem:
\begin{eqnarray}\label{eq_JDRSAM_diffstate_logtrans_BVP_Phi_tilde_kplus1}
        \frac{\partial \tilde{\Phi}^{k}}{\partial t}
        + \frac{1}{2} \textrm{tr}
                    \left( \Lambda \Lambda^T (t) D^2 \tilde{\Phi}^{k}\right)
                +       f(x,h^{k-1})^T  D\tilde{\Phi}^{k} + \theta g(t,x,h^{k-1}) \tilde{\Phi}^{k}
    =   0                                                                                       \nonumber\\
\end{eqnarray}
subject to boundary conditions
\begin{eqnarray}
        \tilde{\Phi}(t, x) &=& \Psi(t,x)
        \qquad \forall (t,x) \in \partial^* Q_R := \left( (0,T) \times \partial \mathscr{B}_R \right) \cup \left( \left\{T\right\} \times \mathscr{B}_R \right)
                                                \nonumber
\end{eqnarray}

Note that the boundary value problem~\eqref{eq_JDRSAM_diffstate_logtrans_BVP_Phi_tilde_kplus1} is a special case of the generic problem introduced earlier in equation~\eqref{eq_JDRSAM_diffstate_generic_linear_PDE} with
\begin{eqnarray}
    a(t,x)  &=& \Lambda\Lambda^T (t)
                                                \nonumber\\
    b(t,x)  &=& f(x,h^{k})
                                                \nonumber\\
    c(t,x)  &=& g(t,x,h^{k})
                                                \nonumber\\
    d(t,x)  &=& 0
                                                \nonumber
\end{eqnarray}
Moreover, since $\mathscr{B}_R$ is bounded and $\overline{\mathcal{J}}$ is compact, all of these functions are also bounded. Thus, based on standard results on parabolic Partial Differential Equations (see for example Appendix E in Fleming and Rishel~\cite{fl75} and Chapter IV in Ladyzhenskaya, Solonnikov and Uralceva~\cite{la68}), the boundary value problem~\eqref{eq_JDRSAM_diffstate_logtrans_BVP_Phi_tilde_kplus1} admits a unique solution in $\mathscr{L}^{\eta}(Q_R)$.
\\

Next, for almost all $(t,x) \in Q_R$, $k=1,2, \ldots$, we define $h^{k}$ by the prescription
\begin{eqnarray}\label{eq_JDRSAM_diffstate_logtrans_PolicyImprovement_def_h_k}
    h^{k} = \textrm{Argmin}_{h \in \overline{\mathcal{J}}}
        \left\{
                    f(x,h)^T  D\tilde{\Phi}^{k}
        +   \theta g(t,x,h) \tilde{\Phi}^{k}
        \right\}
\end{eqnarray}
so that
\begin{eqnarray}\label{eq_JDRSAM_diffstate_logtrans_BVP_Phi_tilde_kplus1_identity}
        f(x,h^{k})^T  D\tilde{\Phi}^{k}
        + \theta g(t,x,h^{k}) \tilde{\Phi}^{k}
    &=& \inf_{h \in \overline{\mathcal{J}}} \left\{
                    f(x,h)^T  D\tilde{\Phi}^{k}
        +   \theta g(t,x,h) \tilde{\Phi}^{k}
        \right\}
                                                \nonumber\\
    &=& H(t,x,\tilde{\Phi}^{k},D\tilde{\Phi}^{k})
\end{eqnarray}
\\

Note, in view of the definition of $g$ in~\eqref{eq_JDRSAM_g_func_def}, that the minimum is never achieved on the boundary of $\overline{\mathcal{J}}$, i.e. $h^k$ takes values in $\mathcal{J}$. 
\\

Observe that the sequence $\left( \tilde{\Phi}^k \right)_{k \in \mathbb{N}}$ is globally bounded. Indeed, by Feynman-K\v ac, the sequence $\left( \tilde{\Phi}^k \right)_{k \in \mathbb{N}}$ is bounded from below by 0. By the optimality principle, it is also bounded from above by $e^{\theta \int_{t}^{T} g(X(s),\check{h})ds} = e^{\theta \check{g} (T-t)}$. Moreover, these bounds do not depend on the radius $R$ and are therefore valid over the entire space $(0,T)\times\mathbb{R}^n$.
\\

\textbf{Step 2: Convergence Inside the Cylinder $(0,T) \times \mathscr{B}_R$}\\

\indent\indent\textbf{Step 2.1: Monotonicity of the Sequence}\\
Take $k \geq 1$. Subtracting the PDE for $\tilde{\Phi}^{k+1}$ from the PDE for $\tilde{\Phi}^{k}$, we see that
\begin{eqnarray}
    &&
        \left( \frac{\partial \tilde{\Phi}^{k+1}}{\partial t} - \frac{\partial \tilde{\Phi}^{k}}{\partial t} \right)
        +   \left(
            \frac{1}{2} \textrm{tr} \left[
                \left( \Lambda \Lambda^T (t) D^2 \tilde{\Phi}^{k+1}\right)
                -   \left( \Lambda \Lambda^T (t) D^2 \tilde{\Phi}^{k}\right)
                \right]\right.
                                                \nonumber\\
    &&  \left.
                +   \left( f(x,h^{k})^T  D\tilde{\Phi}^{k+1} - f(x,h^{k-1})^T  D\tilde{\Phi}^{k} \right)
        +   \theta \left( g(t,x,h^{k}) \tilde{\Phi}^{k+1} - g(t,x,h^{k-1})\tilde{\Phi}^{k} \right)
        \right]
                                               \nonumber\\
    &=&
        0
    \qquad \textrm{in } (0,T)\times\mathbb{R}^n
                                                \nonumber
\end{eqnarray}
with $\tilde{\Phi}^{k+1} - \tilde{\Phi}^{k} = 0$ on $\mathbb{R}^n$.
\\

Add and subtract $f(x,h^{k})^T  D\tilde{\Phi}^{k} + \theta g(t,x,h^{k}) \tilde{\Phi}^{k}$,
\begin{eqnarray}
    &&
        \left( \frac{\partial \tilde{\Phi}^{k+1}}{\partial t} - \frac{\partial \tilde{\Phi}^{k}}{\partial t} \right)
        +   \left(
            \frac{1}{2} \textrm{tr} \left[
                \left( \Lambda \Lambda^T (t) D^2 \tilde{\Phi}^{k+1}\right)
                -   \left( \Lambda \Lambda^T (t) D^2 \tilde{\Phi}^{k}\right)
                \right]\right.
                                                \nonumber\\
    &&
                +   \left( f(x,h^{k})^T  D\tilde{\Phi}^{k+1} - f(x,h^{k-1})^T  D\tilde{\Phi}^{k} \right)
        +   \theta \left( g(t,x,h^{k}) \tilde{\Phi}^{k+1} - g(t,x,h^{k-1})\tilde{\Phi}^{k} \right)
                                               \nonumber\\
        &&      + \left( f(x,h^{k})^T  D\tilde{\Phi}^{k} + \theta g(t,x,h^{k}) \tilde{\Phi}^{k} \right)
                - \left( f(x,h^{k})^T  D\tilde{\Phi}^{k} + \theta g(t,x,h^{k}) \tilde{\Phi}^{k} \right)
                                               \nonumber\\
    &=&
        0
    \qquad \textrm{in } (0,T)\times\mathbb{R}^n
                                                \nonumber
\end{eqnarray}

Rearranging,
\begin{eqnarray}
    &&
        \left( \frac{\partial \tilde{\Phi}^{k+1}}{\partial t} - \frac{\partial \tilde{\Phi}^{k}}{\partial t} \right)
        +   \left(
            \frac{1}{2} \textrm{tr} \left[
                \left( \Lambda \Lambda^T (t) D^2 \tilde{\Phi}^{k+1}\right)
                -   \left( \Lambda \Lambda^T (t) D^2 \tilde{\Phi}^{k}\right)
                \right]\right.
                                                \nonumber\\
    &&
                +   f(x,h^{k})^T \left(D\tilde{\Phi}^{k+1} - D\tilde{\Phi}^{k} \right)
        +   \theta g(t,x,h^{k}) \left( \tilde{\Phi}^{k+1} - \tilde{\Phi}^{k} \right)
                                               \nonumber\\
        &&      + \left( f(x,h^{k})^T  D\tilde{\Phi}^{k} + \theta g(t,x,h^{k})\tilde{\Phi}^{k}\right)
                - \left( f(x,h^{k-1})^T  D\tilde{\Phi}^{k} + \theta g(t,x,h^{k-1})\tilde{\Phi}^{k}\right)
                                               \nonumber\\
    &=&
        0
    \qquad \textrm{in } (0,T)\times\mathbb{R}^n
                                                \nonumber
\end{eqnarray}
\\

Define the function $\ell^k(t,x)$ as
\begin{eqnarray}
        \ell^k(t,x) :=
                \left( f(x,h^{k})^T  D\tilde{\Phi}^{k} + \theta g(t,x,h^{k})\tilde{\Phi}^{k}\right)
                - \left( f(x,h^{k-1})^T  D\tilde{\Phi}^{k} + \theta g(t,x,h^{k-1})\tilde{\Phi}^{k}\right)
                                                                                                                                \nonumber
\end{eqnarray}
By the definition of $h^{k}$ given in~\eqref{eq_JDRSAM_diffstate_logtrans_PolicyImprovement_def_h_k}, $\ell^k(t,x) \leq 0$ $\forall (t,x) \in [0,T]\times\mathbb{R}^n, \forall k \in \mathbb{N}$. Define the sequence of functions $(W^k)_{k \in \mathbb{N}}$ as
\begin{eqnarray}
        W^k := \tilde{\Phi}^{k+1} - \tilde{\Phi}^{k}
                                                                        \nonumber
\end{eqnarray}
then $W^k$ satisfies the PDE
\begin{eqnarray}\label{as_JDRSAM_diffstate_intermediatePDE1}
    &&
        \frac{\partial W^{k}}{\partial t}
      +  \frac{1}{2} \textrm{tr} \left(\Lambda \Lambda^T (t) D^2 W^{k}\right)
                +  f(x,h^{k})^T DW^{k}
      +  \theta g(t,x,h^{k}) W^{k}
                + \ell^k(t,x)
                = 0
                                                                                                                        \nonumber\\
\end{eqnarray}
in $(0,T)\times\mathscr{B}_R$, and with boundary condition $W^{k}(T,x) = 0$ on $\partial^* Q_R = \left( (0,T) \times \partial\mathscr{B}_R \right) \cup \left( \left\{T\right\} \times \mathscr{B}_R \right)$.
\\

Define the stopping time $\tau_G$ as the first exit time from $\mathscr{B}_R$:
\begin{eqnarray}
        \tau_G := \inf\left\{t: X(t) \notin G \right\}
                                                                                                                        \nonumber
\end{eqnarray}
By a standard Feynman-Kac representation, $W^k(t,x)$ can be represented by the expectation
\begin{eqnarray}
        W^k(t,x) = \mathbf{E}\left[ \int_{t}^{T\wedge\tau_G} \ell^k(s,X_s) e^{\theta \int_{0}^{s} g(r,X_r) dr} ds \right]
\end{eqnarray}
Because $\ell(t,x) \leq 0$, $W^k(t,x) \leq 0$ for $k \geq 1$ and hence by definition of $W^k$,
\begin{eqnarray}
        \tilde{\Phi}^{k} \geq \tilde{\Phi}^{k+1}
        , \qquad \forall k \in \mathbb{N}
                                                                        \nonumber
\end{eqnarray}
which implies that the sequence $\left\{ \tilde{\Phi}^{k} \right\}_{k \in \mathbb{N}}$ is non increasing.
\\

\indent\indent\textbf{Step 2.2: Convergence of the Sequence}\\
Since the sequence $(\tilde{\Phi}^k)_{k \in \mathbb{N}}$ is non increasing and is also bounded, it converges. Denote by $\tilde{\Phi}$ its limit as $k \to \infty$. Now, since the Sobolev-type norm $\lVert \tilde{\Phi}^{k+1} \rVert_{\eta,Q_R}^{(2)}$ is bounded for $1 < \eta < \infty$, we can apply the following estimate given by equation (E.9) in Appendix E of Fleming and Rishel
\begin{eqnarray}\label{eq_JDRSAM_diffstate_logtrans_BVP_bound0}
            \lvert \tilde{\Phi}^k \rvert_{Q_R}^{1+\mu}
    \leq    M_R \lVert \tilde{\Phi}^k \rVert_{\eta,Q_R}^{(2)}
\end{eqnarray}
for some constant $M_R$ (depending on $R$) and where
\begin{equation}
    \mu = 1 -\frac{n+2}{\eta}
                                            \nonumber
\end{equation}
\begin{eqnarray}
        \lvert \tilde{\Phi}^k \rvert_{Q_R}^{1+\mu}
    =   \lvert \tilde{\Phi}^k \rvert_{Q_R}^{\mu}
    +   \sum_{i=1}^{n} \lvert \tilde{\Phi}_{x_i}^k \rvert_{Q_R}^{\mu}
                                            \nonumber
\end{eqnarray}
and
\begin{eqnarray}
        \lvert \tilde{\Phi}^k \rvert_{Q_R}^{\mu}
    &=& \sup_{(t,x)\in Q_R} \lvert \tilde{\Phi}^k (t,x) \rvert
    +   \sup_{\begin{array}{c} (x,y)\in \overline{G} \\ 0 \leq t \leq T \end{array}}
            \frac{\lvert \tilde{\Phi}^k(t,x)- \tilde{\Phi}^k(t,y) \rvert}{\lvert x - y \rvert^{\mu}}
                                            \nonumber\\
    &+& \sup_{\begin{array}{c} x \in \overline{G} \\ 0 \leq s,t \leq T \end{array}}
            \frac{\lvert \tilde{\Phi}^k(s,x)- \tilde{\Phi}^k(t,x) \rvert}{\lvert s - t \rvert^{\mu/2}}
                                            \nonumber
\end{eqnarray}
to show that the H\"older-type norm $\lvert \tilde{\Phi}^k \rvert_{Q_R}^{1+\mu}$ is bounded. As $k \to \infty$ we conclude that
\begin{itemize}
\item $D\tilde{\Phi}^k$ converges to $D\tilde{\Phi}$ uniformly in $L^{\eta}(Q_R)$ ;
\item $D^2\tilde{\Phi}^k$ converges to $D^2\tilde{\Phi}$ weakly in $L^{\eta}(Q_R)$ ; and
\item $\frac{\partial \tilde{\Phi}^k}{\partial t}$ converges to  $\frac{\partial \tilde{\Phi}}{\partial t}$ weakly in $L^{\eta}(Q_R)$.
\\
\end{itemize}

\indent\indent\textbf{Step 2.3: Proving that $\tilde{\Phi} \in C^{1,2}(Q_R)$}\\
Using estimate~\eqref{eq_JDRSAM_diffstate_logtrans_BVP_bound0}, we see that $\lvert \tilde{\Phi}^k \rvert_{Q_R}^{1+\mu}$ is bounded for $\mu > 0,$ which implies that $\eta > n+2$. Using relationship~\eqref{eq_JDRSAM_diffstate_logtrans_BVP_Phi_tilde_kplus1_identity} and then equation~\eqref{eq_JDRSAM_diffstate_logtrans_BVP_Phi_tilde_kplus1}, we get:
\begin{eqnarray}\label{eq_JDRSAM_diffstate_logtrans_BVP_ineq1}
    &&
        \frac{\partial \tilde{\Phi}^{k}}{\partial t}
        + \frac{1}{2} \textrm{tr}
                    \left( \Lambda \Lambda^T (t) D^2 \tilde{\Phi}^{k}\right)
                +       f(x,h)^T  D\tilde{\Phi}^{k} + \theta g(t,x,h) \tilde{\Phi}^{k}
                                                \nonumber\\
    &\geq&
        \frac{\partial \tilde{\Phi}^{k}}{\partial t}
        + \frac{1}{2} \textrm{tr}
                    \left( \Lambda \Lambda^T (t) D^2 \tilde{\Phi}^{k}\right)
                +       f(x,h^{k})^T  D\tilde{\Phi}^{k} + \theta g(t,x,h^{k}) \tilde{\Phi}^{k}
                                                \nonumber\\
    &=&
        \left( \frac{\partial \tilde{\Phi}^{k}}{\partial t} - \frac{\partial \tilde{\Phi}^{k+1}}{\partial t} \right)
        +   \left(
            \frac{1}{2} \textrm{tr} \left[
                \left( \Lambda \Lambda^T (t) D^2 \tilde{\Phi}^{k}\right)
                -   \left( \Lambda \Lambda^T (t) D^2 \tilde{\Phi}^{k+1}\right)
                \right]\right.
                                                \nonumber\\
    &&  \left.
                +   f(x,h^{k})^T  \left( D\tilde{\Phi}^{k} - D\tilde{\Phi}^{k+1} \right)
        +   \theta g(t,x,h^{k}) \left( \tilde{\Phi}^{k} - \tilde{\Phi}^{k+1} \right)
        \right]
\end{eqnarray}
for any admissible control $h$.
\\

Since the left-hand side of~\eqref{eq_JDRSAM_diffstate_logtrans_BVP_ineq1} tends weakly in $L^{\eta}(Q_R)$ to
\begin{eqnarray}\label{eq_JDRSAM_diffstate_proofstep2_ineq1}
        \frac{\partial \tilde{\Phi}}{\partial t}
        + \frac{1}{2} \textrm{tr}
                    \left( \Lambda \Lambda^T (t) D^2 \tilde{\Phi} \right)
                +       f(x,h)^T  D\tilde{\Phi} + \theta g(t,x,h) \tilde{\Phi}
\end{eqnarray}
as $k \to \infty$ and the right-hand side tends tends weakly to 0, then we obtain the following inequality
\begin{eqnarray}
        \frac{\partial \tilde{\Phi}}{\partial t}
        + \frac{1}{2} \textrm{tr}
                    \left( \Lambda \Lambda^T (t) D^2 \tilde{\Phi} \right)
                +       f(x,h)^T  D\tilde{\Phi} + \theta g(t,x,h) \tilde{\Phi}
    \geq    0
                                                \nonumber
\end{eqnarray}
almost everywhere in $Q_R$.
\\

Using a measurable selection theorem and following an argument similar to that of Lemma VI.6.1 of Fleming and Rishel~\cite{fl75}, we see that there exists a Borel measurable function $h^{*}$ from $(0,T)\times\mathscr{B}_{R}$ into $\overline{\mathcal{J}}$ (in fact $\mathcal{J}$) such that
\begin{eqnarray}
        f(x,h^*)^T  D\tilde{\Phi}
        + \theta g(t,x,h^*) \tilde{\Phi}
    &=& \inf_{h \in \overline{\mathcal{J}}} \left\{
                    f(x,h)^T  D\tilde{\Phi}
        +   \theta g(t,x,h) \tilde{\Phi}\right\}
                                                                                                                                \nonumber
\end{eqnarray}
holds for almost all $(t,x) \in (0,T)\times\mathscr{B}_{R}$. Then
\begin{eqnarray}\label{eq_JDRSAM_diffstate_logtrans_BVP_ineq2}
    &&
        \frac{\partial \tilde{\Phi}^{k}}{\partial t}
        + \frac{1}{2} \textrm{tr}
                    \left( \Lambda \Lambda^T (t) D^2 \tilde{\Phi}^{k}\right)
                +       f(x,h^{*})^T  D\tilde{\Phi}^{k} + \theta g(t,x,h^{*}) \tilde{\Phi}^{k}
                                                \nonumber\\
    &\leq&
        \frac{\partial \tilde{\Phi}^{k}}{\partial t}
        + \frac{1}{2} \textrm{tr}
                    \left( \Lambda \Lambda^T (t) D^2 \tilde{\Phi}^{k}\right)
                +       f(x,h^{k})^T  D\tilde{\Phi}^{k} + \theta g(t,x,h^{k}) \tilde{\Phi}^{k}
                                                \nonumber\\
    &=&
        \left( \frac{\partial \tilde{\Phi}^{k}}{\partial t} - \frac{\partial \tilde{\Phi}^{k+1}}{\partial t} \right)
        +   \left(
            \frac{1}{2} \textrm{tr} \left[
                \left( \Lambda \Lambda^T (t) D^2 \tilde{\Phi}^{k}\right)
                -   \left( \Lambda \Lambda^T (t) D^2 \tilde{\Phi}^{k+1}\right)
                \right]\right.
                                                \nonumber\\
    &&  \left.
                +   f(x,h^{k})^T  \left( D\tilde{\Phi}^{k} - D\tilde{\Phi}^{k+1} \right)
        +   \theta g(t,x,h^{k}) \left( \tilde{\Phi}^{k} - D\tilde{\Phi}^{k+1} \right)
        \right]
\end{eqnarray}

Since the left-hand side of~\eqref{eq_JDRSAM_diffstate_logtrans_BVP_ineq2} tends weakly in $L^{\eta}(Q_R)$ to
\begin{eqnarray}
        \frac{\partial \tilde{\Phi}}{\partial t}
        + \frac{1}{2} \textrm{tr}
                    \left( \Lambda \Lambda^T (t) D^2 \tilde{\Phi} \right)
                +       f(x,h^{*})^T  D\tilde{\Phi} + \theta g(t,x,h^{*}) \tilde{\Phi}
                                                \nonumber
\end{eqnarray}
as $k \to \infty$ and the right-hand side tends weakly to 0, then we obtain the inequality
\begin{eqnarray}\label{eq_JDRSAM_diffstate_proofstep2_ineq2}
        \frac{\partial \tilde{\Phi}}{\partial t}
        + \frac{1}{2} \textrm{tr}
                    \left( \Lambda \Lambda^T (t) D^2 \tilde{\Phi} \right)
                +       f(x,h^{*})^T  D\tilde{\Phi} + \theta g(t,x,h^{*}) \tilde{\Phi}
    \leq    0
\end{eqnarray}
almost everywhere in $Q_R$.
\\

Combining~\eqref{eq_JDRSAM_diffstate_proofstep2_ineq1} and~\eqref{eq_JDRSAM_diffstate_proofstep2_ineq2}, we have shown that
\begin{eqnarray}
        \frac{\partial \tilde{\Phi}}{\partial t}
        + \frac{1}{2} \textrm{tr}
                    \left( \Lambda \Lambda^T (t) D^2 \tilde{\Phi} \right)
                +       f(x,h^{*})^T  D\tilde{\Phi} + \theta g(t,x,h^{*}) \tilde{\Phi}
    =   0
                                                \nonumber
\end{eqnarray}
almost everywhere in $Q_R$.
\\

Hence, $\tilde{\Phi}$ is a solution of equation~\eqref{eq_JDRSAM_diffstate_exptrans_HJBPDE} on a bounded domain. Moreover, $\tilde{\Phi} \in \mathcal{L}^{\eta}(Q_R)$. Also, since $H$ is locally Lipschitz, $\lvert \tilde{\Phi}^k \rvert_{Q_R}^{\mu} < \infty$ for $\mu > 0$ and $\lvert D\tilde{\Phi}^k \rvert_{Q_R}^{\mu} < \infty$ for $\mu > 0$, then $\lvert H(t,x,\tilde{\Phi}^k,D\tilde{\Phi}^k) \rvert_{Q_R}^{\mu} < \infty$.
\\

We can now show that $\tilde{\Phi} \in C^{1,2}(Q_R)$. Define
\begin{eqnarray}
        \lvert \tilde{\Phi}^k \rvert_{Q_R}^{2+\mu}
   :=   \lvert \tilde{\Phi}^k \rvert_{Q_R}^{1+\mu}
        +       \Big\lvert \frac{\partial \tilde{\Phi}^k}{\partial t} \Big\rvert_{Q_R}^{\mu}
   +    \sum_{i,j=1}^{n} \lvert \tilde{\Phi}_{x_i x_j}^k \rvert_{Q_R}^{\mu}
                                            \nonumber
\end{eqnarray}
Consider the following estimate given by equation (E10) in Appendix E of Fleming and Rishel
\begin{eqnarray}\label{eq_JDRSAM_diffstate_logtrans_BVP_bound1}
            \lvert \tilde{\Phi} \rvert_{Q'}^{2+\mu}
    \leq    M_2 \lVert \tilde{\Phi} \rVert_{Q''}
\end{eqnarray}
for some constant $M_2$, and two open subsets $Q'$ and $Q''$ of $Q$ such that $\bar{Q'} \subset \bar{Q''}$. In this estimate, set $Q'' = Q_R$ and take $Q'$ to be any subset of $Q$ such that $\bar{Q}' \subset Q$. Thus
\begin{eqnarray}\label{eq_JDRSAM_diffstate_logtrans_BVP_bound2}
            \lvert \tilde{\Phi} \rvert_{Q'}^{2+\mu} < \infty
\end{eqnarray}
When interpreted in light of estimate~\eqref{eq_JDRSAM_diffstate_logtrans_BVP_bound0} (stemming from (E9)), we see that the derivatives $\frac{\partial \tilde{\Phi}}{\partial t}$, $\frac{\partial \tilde{\Phi}}{\partial x_i}$ and $\frac{\partial^2 \tilde{\Phi}}{\partial x_i x_j}$ satisfy a uniform H\"older condition on any compact subset $Q'$ of $Q_R$. By Theorem 10.1 in Chapter IV of Ladyzhenskaya, Solonnikov and Uralceva~\cite{la68}, we can therefore conclude that $\tilde{\Phi} \in C^{1,2}(Q_R)$.
\\

\textbf{Step 3: Convergence from the Cylinder $[0,T) \times \mathscr{B}_R$ to the State Space $[0,T) \times \mathbb{R}^n$}\\

\indent\indent \textbf{Step 3.1: Setting}
\\
Let $\left\{R_i\right\}_{i \in \mathbb{N}} > 0$ be a non decreasing sequence with $\lim_{i \to \infty} R_i \to \infty$ and let $\left\{ \tau_i \right\}_{i \in \mathbb{N}}$ be the sequence of stopping times defined as
\begin{eqnarray}
        \tau_i := \inf\left\{t : X(t) \notin \mathscr{B}_{R_i} \right\}
                                                                                                                                \nonumber
\end{eqnarray}
Note that $\left\{ \tau_i \right\}_{i \in \mathbb{N}}$ is non decreasing and $\lim_{i \to \infty} \tau_i = \infty$.
\\

Denote by $\tilde{\Phi}^{(i)}$ the limit of the sequence $\left(\tilde{\Phi}^k \right)_{k\in \mathbb{N}}$ on $(0,T)\times\mathscr{B}_{R_i}$, i.e.
\begin{eqnarray}
        \tilde{\Phi}^{(i)}(t,x) = \lim_{k \to \infty} \tilde{\Phi}^k(t,x)
                \qquad \forall (t,x) \in (0,T)\times\mathscr{B}_{R_i}
\end{eqnarray}
\\

\indent\indent \textbf{Step 3.2: Convergence of the sequence $\left( \tilde{\Phi}^{(i)} \right)_{i \in \mathbb{N}}$}
\\
First, observe that the sequence $(\tilde{\Phi}^{(i)})_{i \in \mathbb{N}}$ is non increasing. Indeed, for $i < j$ the stochastic control problem defined over $(0,T)\times\mathscr{B}_{R_i}$ is nested into the stochastic control problem defined over $(0,T)\times\mathscr{B}_{R_j}$. In particular, a suboptimal strategy for the stochastic control problem defined over $(0,T)\times\mathscr{B}_{R_j}$ would be to invest optimally while $x \in \mathscr{B}_{R_i}$ and then switch to the $0\beta$ policy $\check{h}$ when $x \in \mathscr{B}_{R_j}\backslash\mathscr{B}_{R_i}$. By the optimality principle, the expected total cost of such strategy is greater than the value function $\tilde{\Phi}^{(j)}$. But this suboptimal strategy also corresponds to the optimal strategy for the stochastic control problem defined over $(0,T)\times\mathscr{B}_{R_i}$. Hence
\begin{eqnarray}
        \tilde{\Phi}^{(i)}(t,x) \geq \tilde{\Phi}^{(j)}(t,x)
        \qquad \forall i,j \in \mathbb{N},
        \; \forall (t,x) \in (0,T)\times\mathscr{B}_{R_i}
                                                                                                \nonumber
\end{eqnarray}

By the argument in Step 1, the sequence $(\tilde{\Phi}^{(i)})_{i \in \mathbb{N}}$ is also bounded. As a result, it converges to a limit $\tilde{\Phi}$. This limit satisfies the boundary condition~\eqref{eq_JDRSAM_diffstate_exptrans_HJBPDE_termcond}. We now show that $\tilde{\Phi}$ is $C^{1,2}$ and satisfies the HJB PDE. These statements are local properties so we can restrict ourselves to a finite ball $Q_R$.
\\

\indent\indent \textbf{Step 3.3: Proving that $\tilde{\Phi} \in C^{1,2}(Q_R)$}
\\
Now that we have shown the convergence of the sequence $\left( \tilde{\Phi}^{(i)} \right)_{i \in \mathbb{N}}$ though a simple control-based argument, we can conclude the proof using the same arguments based on Ascoli's theorem as Fleming and Rishel (see \cite{fl75}, proof of Theorem 6.2 in Appendix E).


\end{proof}

\begin{corollary}[Existence of a Classical Solution for the Risk-Sensitive Control Problem]\label{Coro_JDRSAM_diffstate_existence}
        The RS HJB PDE~\eqref{eq_JDRSAM_diffstate_HJBPDE} with terminal condition $\Phi(T,x) = \ln v$ has a solution $\Phi \in C^{1,2}\left([0,T]\times\mathbb{R}^n\right)$ with $\Phi$ continuous in $\overline{[0,T]\times\mathbb{R}^n}$.

\end{corollary}

%
%

\section{Partial Observation}
In this section we show how the results of the paper can be extended to the case where the factor process $X(t)$ is not directly observed and the asset allocation strategy $h_t$ must be adapted to the filtration $\cF^S_t=\sigma\{S_i(u),0\leq u \leq t, j=0,\ldots,m\}$ generated by the asset price processes alone. In the linear diffusion case studied by Nagai~\cite{na00} and Nagai and Peng \cite{nape02}, the authors noted that the pair of processes $(X(t),Y(t))$, where $Y_i(t)=\log S_i(t)$, take the form of the `signal' and `observation' processes in a Kalman filter system, and consequently the conditional distribution of $X(t)$ is normal $N(\hat{X}(t),P(t))$ where $\hat{X}(t)=\mathbf{E}[X(t)|\cF^S_t]$ satisfies the Kalman filter equation and $P(t)$ is a deterministic matrix-valued function. By using this idea they obtain an equivalent form of the problem in which $X(t)$ is replaced by $\hat{X}(t)$ and the dynamic equation \eqref{eq_FactorProcess_diffusion} by the Kalman filter. Optimal strategies take the form $h(t,\hat{X}(t))$. This is in fact a very old idea in stochastic control, going back at least to Wonham \cite{w}.

\subsection{Decomposition}
At first sight it does not seem apparent that the same approach can be used here, as the price processes contain jumps, but a simple observation shows that the jumps play no role in the estimation process, which is still, at base, the Kalman filter; see Proposition \ref{1+2} below. A further complication is that the money market interest rate $r(t)=a_0+A_0^T X(t)$ (see \eqref{eq_JDRSAM_BankAccount}) is observed directly and contains information about $X(t)$. This was not the case in \cite{na00} and \cite{nape02} where, in our notation, $A_0=0$.   We start by assuming that $A_0=0$, and briefly discuss the extension to $A_0\not=0$ at the end of the section.

Recall first that $X(t)$ satisfies
\begin{equation}\label{X}
    dX(t) = (b + BX(t))dt + \Lambda dW(t),
    \qquad X(0) = X_0
\end{equation}
When $X_t$ is observed, the initial value $X_0$ is just a constant. In the present case we need to assume that $X_0$ is a normal random vector $N(m_0, P_0)$ with known mean $m_0$ and covariance $P_0$, and that $X_0$ is independent of the processes $W, N_{\mathbf{p}}$.

An application of the general It\^o formula\footnote{See \O ksendal and Sulem \cite{oksu05} for this calculation.} shows that for $i=1,\ldots,m$ the log-prices $Y_i(t)$ satisfy $Y_i(0)=\log s_i$ and
\begin{eqnarray}\label{eq_ObservationProcess}
    dY_i(t) &=&
        \left[(\hat{a} + \hat{A}X(t))_{i} -\frac{1}{2} \Sigma\Sigma_{ii}^T  \right]dt
        + \sum_{k=1}^{N} \sigma_{ik} dW_{k}(t) \nonumber\\
      &  + &\int_{\mathbf{Z}_0} \left\{  \ln\left(1+\gamma_i(z)\right) - \gamma_i(z) \right\} \nu(dz)dt
        + \int_{\mathbf{Z}} \ln \left( 1+\gamma_i(z) \right) \bar{N}(dt,dz).
\end{eqnarray}

\begin{proposition}\label{1+2}
Define processes $Y^1(t), Y^2(t)\in\mathbb{R}^m$ as follows.
\begin{eqnarray}\label{eq_ObservationProcess_stdform}
        dY^1(t)& = &\hat{A}X(t) + \Sigma dW(t),\qquad\qquad\qquad\qquad Y^1_i(0)=0,   \\
        dY_i^2(t) &= & c_i dt +  \int_{\mathbf{Z}} \ln\left( 1+\gamma(z) \right)_i \bar{N}(dt,dz),
        \qquad i=1,\ldots,m, \qquad Y^2_i(0)=\log s_i
                                                \nonumber
\end{eqnarray}
with $c\in\mathbb{R}^m$ defined by
\begin{eqnarray}
        c_i :=  \hat{a}_i
        -       \frac{1}{2} \Sigma\Sigma_{ii}^T 
        +       \int_{\mathbf{Z}_0} \left\{  \ln\left(1+\gamma_i(z)\right) - \gamma_i(z) \right\} \nu(dz)
                                                \nonumber
\end{eqnarray}
so that $Y(t)=Y^1(t)+Y^2(t)$. Also, define $\cY_{it}=\sigma\{Y^i(u),0\leq u\leq t\},\,i=1,2$. Then

\noindent(i) The processes $Y^1, Y^2$ are each adapted to the filtration $\cF^S_t$.

\noindent(ii) For any bounded measurable function $f$ and $t\geq 0$,
\[ \bE[f(X(t))|\cF^S_t]=\bE[f(X(t))|\cY_{1t}].\]
\end{proposition}
\proof (i) $S(t)$ and $Y(t)$ are in 1-1 correspondence and therefore generate the same filtration $\cF^S_t$. Apart from rearrangement of deterministic terms, the decomposition $Y=Y^1+Y^2$ is the same as the standard  decomposition $Y=Y^c+Y^d$ of a semimartingale into its continuous and discontinuous components, see paragraph VI. 37 of Rogers and Williams \cite{rw2}.

\noindent (ii) $N$ and $W(t)$ are independent and as a result $\cY_{1t}$ and $\cY_{2t}$ are independent, and clearly $\cF^S_t=\cY_{1t}\vee\cY_{2t}$. The result follows, since $X(t)$ is independent of $\cY_{2t}$.\hfill$\square$

\subsection{Kalman Filter}
The processes $(X(t),Y^1(t))$ satisfying \eqref{X} and \eqref{eq_ObservationProcess_stdform} and the filtering equations, which are standard, are stated in the following proposition.
\begin{proposition}[Kalman Filter]
The conditional distribution of $X(t)$ given $\cY_{1t}$ is $N( \hat{X}(t), P(t))$, calculated as follows.

\noindent(i) The \emph{innovations process} $U(t)\in\bR^m$ defined by
\begin{equation}\label{eq_InnovationProcess}
        dU(t)=\left(\Sigma\Sigma^T \right)^{-1/2}(dY^1(t)-\hat{A}\hat{X}(t)dt),\qquad U(0)= 0
\end{equation}
is a vector Brownian motion.

\noindent (ii) $\hat{X}(t)$ is the unique solution of the SDE
\begin{equation}\label{eq_FactorProcess_diffusion_Estimation_Final}
    d\hat{X}(t)
    = (b+  B \hat{X}(t)) dt
    +   \left( \Lambda\Sigma^T  + P(t)\hat{A}^T  \right)\left(\Sigma\Sigma^T \right)^{-1/2} dU(t),
    \qquad \hat{X}(0) = m_0.
\end{equation}

\noindent(iii) $P(t)$ is the unique non-negative definite symmetric solution of the matrix Riccati equation
\begin{eqnarray}
        \dot{P}(t)
        &=&     \Lambda\Xi\Xi^T \Lambda^T 
        -       P(t)\hat{A}^T \left(\Sigma\Sigma^T \right)^{-1} \hat{A}P(t)
        +       \left( B - \Lambda\Sigma^T \left(\Sigma\Sigma^T \right)^{-1}\hat{A}\right)P(t)
                                                        \nonumber\\
        &&
        +       P(t)\left( B^T  - \hat{A}^T \left(\Sigma\Sigma^T \right)^{-1}\Sigma\Lambda^T \right),
        \qquad  P(0) = P_0
                                                        \nonumber
\end{eqnarray}
where $\Xi := I - \Sigma^T \left(\Sigma^T \Sigma\right)^{-1}\Sigma$.
\newline
\end{proposition}

To conclude, the Kalman filter has replaced our initial state process $X(t)$ by an estimate $\hat{X}(t)$ with dynamics given in~\eqref{eq_FactorProcess_diffusion_Estimation_Final}. To recover the asset price process,  we use~\eqref{eq_ObservationProcess_stdform} together with~\eqref{eq_InnovationProcess} to obtain the dynamics of $Y(t)$:
\begin{eqnarray}
        dY_i(t)
        &=& dY_i^1(t) + dY_i^2(t)
                                                                        \nonumber\\
        &=&     \hat{a}_i + \hat{A} \hat{X}(t) dt
        -       \frac{1}{2} \Sigma\Sigma_{ii}^T  dt
        +       \left(\Sigma\Sigma^T \right)^{1/2} dU(t)
                                                                        \nonumber\\
        &&
        +       \int_{\mathbf{Z}_0} \left\{  \ln\left(1+\gamma_i(z)\right) - \gamma_i(z) \right\} \nu(dz)
        +       \int_{\mathbf{Z}} \ln\left( 1+\gamma(z) \right)_i \bar{N}(dt,dz).
\end{eqnarray}
We then apply It\^o to $S_i(t) =  \exp Y_i(t)$ to get
\begin{eqnarray}\label{eq_SecurityProcess_Estimation_Final}
    \frac{dS_{i}(t)}{S_{i}(t^{-})} &=&
        (a + A\hat{X}(t))_{i}dt
        + \sum_{k=1}^{N} \left[\left(\Sigma\Sigma^T \right)^{1/2}\right]_{ik} dU_{k}(t)
        + \int_{\mathbf{Z}} \gamma_i(z)\bar{N}(dt,dz),
                                                \nonumber\\
    && S_i(0) = s_i,
    \quad i = 1, \ldots, m
\end{eqnarray}
\newline

We now solve the stochastic control problem with partial observation simply by replacing the original asset price description \eqref{eq_SecurityProcess} by \eqref{eq_SecurityProcess_Estimation_Final}, and the factor process description \eqref{eq_FactorProcess_diffusion} by the Kalman filter equation \eqref{eq_FactorProcess_diffusion_Estimation_Final}, in our solution of full observation case. The Kalman filter has time-varying coefficients, but this does not affect the preceding arguments.

Finally, we briefly sketch what to do if $A_0\not=0$. We observe the short rate $r(t)=a_0+A_0^T X(t)$, and hence the 1-dimensional statistic $Y_0(t)\equiv A_0^T X(t)$, exactly. We need to assume that this observation contains positive `noise', i.e. $A_0^T \Lambda\Lambda^T A_0>0$. Changing coordinates if necessary, we can assume that $A_0^T  = (0,0,\ldots,1)$ and hence $Y_0(t)=X_n(t)$. Our `observation' is now the $(m+1)$-dimensional process $\bar{Y}=(Y_0,\ldots, Y_m)$ and we can set up a Kalman filter system to estimate the unobserved states $\bar{X}=(X_1,\ldots,X_{n-1})^T \in \bR^{n-1}$. Ultimately, our optimal strategy will take the form $h(t,X_1(t),\hat{\bar{X}}(t))$, where $\hat{\bar{X}}(t)$ is the Kalman filter estimate for $\bar{X}(t)$ given $\{\bar{Y}(u), u\leq t\}$. The details are left to the reader.

%
%

\section{Conclusion}
In this article, we extended the classical risk-sensitive asset management setting to include the possibility of infinite activity jumps in asset prices. We applied the change of measure technique proposed by Kuroda and Nagai~\cite{kuna02} to derive the Hamilton-Jacobi-Bellman Partial Differential Equation associated with the control problem and then proved the existence and uniqueness of an admissible optimal control policy. Using an approximation in policy space algorithm, we established the existence of a classical $C^{1,2}\left((0,T)\times\mathbb{R}^n\right)$ solution and obtained the uniqueness of this solution through a verification result. This approach also extends naturally and with similar results to a jump-diffusion version of the risk-sensitive benchmarked asset management problem.
\\

Finally, we have observed that an attractive, if somewhat surprising, feature of the jump diffusion risk sensitive asset management is that it naturally prohibits any investment policy which may result in the investor's bankruptcy.  In particular, in the risk-sensitive setting presented in this article, an investor who implements the optimal asset allocation is certain of remaining solvent over the investment horizon. This contrasts with the Merton type of approach in which the threat of bankruptcy remains present and has to be accounted for using a stopping time.
\\

The approach presented in this article extends naturally to a jump-diffusion version of the risk-sensitive benchmarked asset management problem introduced by Davis and Lleo~\cite{dall_RSBench} and would yield similar results, namely the existence of a unique admissible control policy and of a classical $C^{1,2}$ solution to the associated RS HJB PDE.
\\

%
%

\appendix

\section{Admissibility of the Optimal Control Policy}\label{sec_Admissibility}
The admissibility of the optimal control process $h^*(t)$ solving~\eqref{eq_JDRSAM_diffstate_supL_deriv} is linked to the existence of a probability measure $\mathbb{P}_{h^{*}}^{\theta}$, which itself hinges on the characterisation as an exponential martingale of the Radon-Nikodym derivative $\frac{d\mathbb{P}_{h^*}^{\theta}}{d\mathbb{P}} = \chi_T^{*}$ defined in~\eqref{eq_JDRSAM_RNder_chi} via the Dol\'eans exponential introduced in~\eqref{eq_JDRSAM_Doleansexp_chi}. In the setting of Kuroda and Nagai~\cite{kuna02}, the admissibility of the control follows easily from an argument in Gihman and Skhorokhod~\cite{gisk72} which proves that the the Dol\'eans exponential (here a Girsanov exponential with Gaussian integrand) is an exponential martingale. However, when the Dol\'eans exponential does not have continuous path, as is the case in a jump diffusion setting, proving that it is indeed a martingale is more difficult. As noted by Protter~\cite{Pr05}, some partial results exist in  this case (see for example M\'emin~\cite{me79} and more recently Protter and Shimbo~\cite{prsh08}), but none is as powerful as their counterparts in the continuous case, namely the Kamazaki or the Novikov conditions.
\\

To show that the Dol\'eans exponential introduced in~\eqref{eq_JDRSAM_Doleansexp_chi} is a martingale we will apply results derived by M\'emin~\cite{me79}. We recall here the definition of the Dol\'eans-Dade exponential as well as results from M\'emin~\cite{me79} (see also Exercise 13 in Chapter V of~\cite{Pr05}) on the multiplicative decomposition of local martingales that we will use to prove our point.
\\

\begin{definition}[Dol\'eans-Dade exponential]\label{def_DDexponential}
        The Dol\'eans-Dade exponential $\mathcal{E}(X)(t)$ of a semimartingale $X$(t) is defined as
\begin{eqnarray}\label{def_JDRSAM_DoleansDade_exponenial}
        \mathcal{E}(X)(t)
        =       \exp\left\{ X(t) - \frac{1}{2}\left[X^c,X^c\right]_t \right\}
                \prod_{0 < s \leq t} (1 + \Delta X_t)e^{- \Delta X_s}
\end{eqnarray}

\end{definition}

\begin{definition}[M\'emin's Additive Decomposition of Local Martingales]\label{def_Memin_decomposition_additive}

Let $M(t)$ be a local martingale.  We define an additive decomposition of $M$ into two processes $M_1(t)$ and $M_2(t)$, i.e. such that $M(t) = M_1(t) + M_2(t)$.
\\

In this decomposition, the process $M_1(t)$ is defined as $M_1(t) = L(t) - \tilde{L}(t)$ where
\begin{eqnarray}
        L(t) = \sum_{0 < s \leq t} \Delta M_s \mathit{1}_{\left\{ |\Delta M_s| \geq \frac{1}{2}   \right\}}
                                                                                                        \nonumber
\end{eqnarray}
and $\tilde{L}(t)$ is the compensator of $L(t)$.

\end{definition}

\begin{proposition}[M\'emin's Proposition III-1]\label{prop_Memin_III1}
Let $M(t)$ be a local martingale with additive decomposition as per definition~\ref{def_Memin_decomposition_additive} and such that $M_0 = 0$. Then
\begin{enumerate}[(i)]
\item $\mathcal{E}(M)$ has the decomposition
                \begin{equation}
                                \mathcal{E}(M) = \mathcal{E}(M_2)\mathcal{E}(\tilde{M}_1)
                                                                                                                \nonumber
                \end{equation}
                where
                \begin{equation}
                                \tilde{M}_1(t) = M_1(t) - \sum_{0 < s \leq t} \frac{\Delta M_1(s)\Delta M_2(s)}{1+\Delta M_2(s)}, \qquad t < \infty
                                                                                                                \nonumber
                \end{equation}
\item $\mathcal{E}(M_2)\tilde{M}_1$ is a local martingale.
\item If $\Delta M(s) > -1$ then $\Delta \tilde{M}_1(s) > -1$ for all finite $s$.

\end{enumerate}

\end{proposition}

\begin{corollary}[M\'emin's Corollary III-2]\label{coro_Memin_III1}
        Let $N$ be a local martingale such that $\Delta N(s) > -1$ for all finite $s$, and such that $\mathcal{E}(N(\infty)$ is uniformly integrable. Let $\mathbb{P}'$ be the probability defined as
\begin{eqnarray}
        \frac{d\mathbb{P}'}{d\mathbb{P}} = \mathcal{E}(N)(\infty)
                                                                                                                \nonumber
\end{eqnarray}

Let $N_1$ be a local martingale with locally integrable variations and denote by $\tilde{N}_1$ the $\mathbb{P}$-semimartingale defined as
\begin{equation}
        \tilde{N}_1(t) = N_1(t) - \sum_{0 < s \leq t} \frac{\Delta N_1(s)\Delta N(s)}{1+\Delta N(s)}, \qquad t < \infty
                                                                                                                \nonumber
\end{equation}
then $\tilde{N}_1$ is a $\mathbb{P}'$ local martingale, with locally integrable variations. Moreover, the $\mathbb{P}'$ predictable compensator of $\sum_{0 < s \leq t} |\Delta \tilde{N}_1(s) |$ is equal to the $\mathbb{P}$ predictable compensator of $\sum_{0 < s \leq t} |\Delta N_1(s) |$.
\\
\end{corollary}

\begin{theorem}[M\'emin's Theorem III-3]\label{theo_Memin_III3}
        Let $M(t)$ be a local martingale with additive decomposition as per definition~\ref{def_Memin_decomposition_additive}. If the predictable compensator of the process
\begin{eqnarray}\label{def_Y_t_process}
                Y(t)
        =       \left[M^c,M^c\right]_t
        +       \sum_{0 < s \leq t} | \Delta M_1(s) |
        +       \sum_{0 < s \leq t} \left( \Delta M_2(s) \right)^2
\end{eqnarray}
is bounded, then $\mathcal{E}(M)(t)$ is uniformly integrable.
\\
\end{theorem}
%
%

\textit{Proof of Proposition~\ref{prop_JDRSAM_diffstate_hstar_admissible}}.
To prove that the control $h^*(t)$ is admissible, we need to show that the local martingale $M^{*}(t)$ defined as
\begin{eqnarray}
        M^{*}(t) :=
        - \theta \int_{0}^{t} \left(h^*(s)\right)^T \Sigma dW_s
        - \int_{0}^{t} \int_{\mathbf{Z}} \ln\left(1-G(z,h^*(s))\right)\tilde{N}(ds,dz)
\end{eqnarray}
and such that
\begin{equation}
        \mathcal{E}(M)(t) = \chi_t^{*}
                                                                                                        \nonumber
\end{equation}
is an exponential martingale.
\newline

To achieve this objective, we will define a new class of control processes to which the optimal control belongs. We will start from the definition of a control $h$ as a function:
\begin{eqnarray}
        h: [0,T]\times\mathbb{R}^n              &       \to     & \mathcal{J}
                                                                                                \nonumber\\
                (t,x)                                                           & \mapsto & h(t,x)
                                                                                                \nonumber
\end{eqnarray}
where the set $\mathcal{J}$ was defined in~\eqref{def_JDRSAM_setJ}. Based on this definition, the control space can be viewed as a functional space.
\newline

Define the functional $\mathcal{L}(x,p,h)$ as
\begin{eqnarray}\label{eq_JDRSAM_mathcalL}
        \mathcal{L}(x,p,h)
        &:=&
            - \frac{1}{2} \left(\theta+1 \right)h^T \Sigma\Sigma^T h
            -\theta h^T \Sigma\Lambda^T p
            +h^T (\hat{a} + \hat{A}x)
                                                                                \nonumber\\
            &&
            -\frac{1}{\theta}\int_{\mathbf{Z}}\left\{
                                                \left[\left(1+h^T \gamma(z)\right)^{-\theta}-1\right]
                                        +\theta h^T \gamma(z)\mathit{1}_{\mathbf{Z}_0}(z)
                    \right\}\nu(dz)
                                                                                                        \nonumber\\
\end{eqnarray}
where $p \in \mathbb{R}^n$ so that
\begin{eqnarray}
        \sup_{h \in \mathcal{J}} L_{t}^{h}\Phi
        &=&
            \left( b+ Bx \right)^T D\Phi
            + \frac{1}{2} \textrm{tr} \left( \Lambda \Lambda^T  D^2 \Phi \right)
            - \frac{\theta}{2} (D\Phi)^T \Lambda \Lambda^T  D\Phi
            +a_0+A_0^T x
                                                                        \nonumber\\
            &&
            +\sup_{h \in J} \mathcal{L}(x,D\Phi,h)
                                                                                                \nonumber\\
\end{eqnarray}
and the unique maximizer of $L_t^h \Phi (t,x)$, $\hat{h}(t,x)$, is also the unique maximizer of $\mathcal{L}(x,D\Phi,h)$.
\newline

Observe that with the choice of control function $h^0(t,x) := 0$ $\forall (t,x) \in [0,T] \times \mathbb{R}^n$, the functional $\mathcal{L}(x,p,h^0) = 0$ $\forall (t,x,p) \in [0,T] \times \mathbb{R}^n \times \mathbb{R}^n$. Invoking the optimality principle, we deduce that $\mathcal{L}(x,D\Phi,h^*(t,x)) \geq 0$.
\newline

Denote by $\hat{\mathcal{J}}$ the range of the control functions $\tilde{h}(t,x)$ such that $\mathcal{L}(x,p,\tilde{h}) \geq 0$. Under Assumption~\ref{as_assetjumps_upanddown_1}, the set $\mathcal{J}$, defined by~\eqref{def_JDRSAM_setJ}, is in the interior of a hypercube and since the functional $\mathcal{L}(x,p,h)$ is smooth, strictly concave in $h$ and $\lim_{h \to \partial \mathcal{J}} \mathcal{L}(x,p,h) = -\infty$, we deduce that the set $\hat{\mathcal{J}}$ is a closed convex subset of $\mathcal{J}$ for all $(t,x) \in [0,T] \times \mathbb{R}^n$. The control functions $\tilde{h}$ take the form
\begin{eqnarray}
        \tilde{h}: [0,T]\times\mathbb{R}^n                &       \to     & \hat{\mathcal{J}} \subset \mathcal{J}
                                                                                                \nonumber\\
                (t,x)                                                           & \mapsto & \tilde{h}(t,x)
                                                                                                \nonumber
\end{eqnarray}

More formally, we can define a class $\hat{\mathcal{H}}(T)$ of Markov control processes as
\\

\begin{definition}\label{def_JDRSAM_controlprocess_hat_h}
    A control process $\tilde{h}(t)$ is in class $\hat{\mathcal{H}}(T)$ if the
    following conditions are satisfied:
    \begin{enumerate}
        \item $\tilde{h}(t)$ is in class $\mathcal{H}$ introduced in Definition~\ref{def_JDRSAM_controlprocess_h};

        \item $\tilde{h}(t,x) \in \hat{\mathcal{J}}$ $\forall (t,x) \in [0,T]\times\mathbb{R}^n$.

    \end{enumerate}
\end{definition}

In particular, we note that the optimal control process $h^{*}(t) \in \hat{\mathcal{H}}(T)$ $\forall t \in [0,T]$ and $\forall \omega \in \Omega$.
\newline

For any control policy $\tilde{h}(t) \in \hat{\mathcal{H}}(T)$, define the local martingale $\hat{M}(t)$ as
\begin{eqnarray}\label{def_JDRSAM_mart_hatMt}
        \hat{M}(t) :=
        - \theta \int_{0}^{t} \tilde{h}(s)^T \Sigma dW_s
        - \int_{0}^{t} \int_{\mathbf{Z}} \ln\left(1-G(z,\tilde{h}(s))\right)\tilde{N}(ds,dz)
\end{eqnarray}

Also, let $L(t)$ be the process defined as
\begin{eqnarray}
        L(t)    &=& \sum_{0 < s \leq t} \Delta Y_s \mathit{1}_{\left\{ |\Delta Y_s| \geq \frac{1}{2}   \right\}}
                                                                                                        \nonumber\\
                        &=& - \int_{0}^{t} \int_{\mathbf{Z} \backslash \mathbf{Z}_1} \ln\left(1-G(z,\tilde{h}(s))\right)N(ds,dz)
                                                                                                        \nonumber
\end{eqnarray}
where $\mathbf{Z}_1 = \left\{ z \in \mathbf{Z}: |\Delta Y_s| < \frac{1}{2},  0 \leq s \leq t \right\}$. Then, the process $M_1(t) := L(t) - \tilde{L}(t)$ can be expressed as:
\begin{eqnarray}
        M_1(t)  &=& - \int_{0}^{t} \int_{\mathbf{Z} \backslash \mathbf{Z}_1} \ln\left(1-G(z,\tilde{h}(s))\right)\tilde{N}(ds,dz)
                                                                                                        \nonumber
\end{eqnarray}
To complete our decomposition of the local martingale $M(t)$, we define the process $M_2(t)$ as
\begin{eqnarray}
        M_2(t)
        &=& M(t) - M_1(t)
                                                                                                        \nonumber\\
        &=& - \int_{0}^{t} \int_{\mathbf{Z}_1} \ln\left(1-G(z,\tilde{h}(s))\right)\tilde{N}(ds,dz)
                                                                                                        \nonumber
\end{eqnarray}

The next step is to study each component of the process $Y(t)$ defined in~\eqref{def_Y_t_process}:
\begin{itemize}
\item The process
        \begin{eqnarray}
                \left[M^c,M^c\right]_t  = \exp \left\{ \theta^2 \int_{0}^{t} \tilde{h}(s)^T \Sigma\Sigma^T \tilde{h}(s) ds \right\}
                                                                                                \nonumber
        \end{eqnarray}
        is clearly bounded because $\tilde{h}(s) \in \hat{\mathcal{H}}(T)$ for all $s \in [0,t]$;

\item The process
        \begin{eqnarray}
                \sum_{0 < s \leq t} | \Delta M_1(s) |
                = \int_{0}^{t} \int_{\mathbf{Z} \backslash \mathbf{Z}_1} \Big| \ln\left(1-G(z,\tilde{h}(s))\right) \Big| N(ds,dz)
                                                                                                \nonumber
        \end{eqnarray}
        is bounded because $\tilde{h}(s) \in \hat{\mathcal{H}}(T)$ for all $s \in [0,t]$. In addition, the number of jumps greater than $\frac{1}{2}$ is finite:
        \begin{eqnarray}
                        \# \left\{0 \leq s \leq t; | \Delta M_1(s) | \right\}
                &=&\#\left\{0 \leq s \leq t; | \Delta M(s) | \mathit{1}_{\left\{ |\Delta M_s| \geq \frac{1}{2}   \right\}} \right\}
                                                                                                                                        \nonumber\\
                &=&N\left(t, \left] -\infty,-\frac{1}{2} \right[ \cup \left]\frac{1}{2},\infty \right[\right)
                                                                                                                                        \nonumber\\
                &<&\infty
                                                                                                                                        \nonumber
        \end{eqnarray}

\item Finally, we turn our attention to the process
        \begin{eqnarray}
                \sum_{0 < s \leq t} \left( \Delta M_2(s) \right)^2
                = \int_{0}^{t} \int_{\mathbf{Z}_1} \Big| \ln\left(1-G(z,\tilde{h}(s))\right) \Big|^2 N(ds,dz)
                                                                                                \nonumber
        \end{eqnarray}
        Recalling that we assumed that in our setting
        \begin{equation}
    \int_{\mathbf{Z}_0} \lvert\gamma(z)\rvert^2 \nu(dz) < \infty
                                            \nonumber
        \end{equation}
        and taking into consideration the fact that $\tilde{h}(s) \in \hat{\mathcal{H}}(T)$ for all $s \in [0,t]$, then we deduce that
        \begin{equation}
    \int_{\mathbf{Z}_0} \Big| \ln\left(1-G(z,\tilde{h}(s))\right) \Big|^2 \nu(dz) < \infty
                                            \nonumber
        \end{equation}
        for any $\omega \in \Omega$, which proves that
                \begin{equation}
    \int_{0}^{t} \int_{\mathbf{Z}_1} \Big| \ln\left(1-G(z,\tilde{h}(s))\right) \Big|^2 N(ds,dz) < \infty
                                            \nonumber
        \end{equation}
\end{itemize}

By Theorem~\ref{theo_Memin_III3}, the Dol\'eans-Dade exponential
\begin{equation}
        \mathcal{E}(\hat{M})(t) = \chi_t^{*}
                                                                                                        \nonumber
\end{equation}
is uniformly integrable for all $\tilde{h} \in \hat{\mathcal{H}}(T)$. We can now apply Corollary~\ref{coro_Memin_III1} to formally define the measure $\mathbb{P}_{\tilde{h}}^{\theta}$. In particular, the measure $\mathbb{P}_{h^{*}}^{\theta}$ characterized via the Radon-Nikodym derivative $\chi_t^{*}$ is well defined because $h^{*}(t) \in \hat{\mathcal{H}}(T)$ $\forall \omega \in \Omega$. This proves that the control $h^*(t)$ is admissible for all $t \in [0,T]$ and $\omega \in \Omega$.

\vbox{\hrule height0.6pt\hbox{%
   \vrule height1.3ex width0.6pt\hskip0.8ex
   \vrule width0.6pt}\hrule height0.6pt
  }

Note that the control policy $h^0(t) = 0$ corresponds to investing the entire wealth into the money market asset for the duration of the investment period. The associated measure $\mathbb{P}_{h^{0}}^{\theta}$ is well defined and it is equal to the physical measure $\mathbb{P}$. In fact, this proof not only shows that the optimal control process $h^*(t)$ is admissible, but also that a large class of ``reasonable'' control processes $\tilde{h}(t)$ is also admissible and is associated with a well-defined probability measure.
\\

\textit{Proof of Proposition~\ref{prop_JDRSAM_diffstate_OptimalJandI}}. Consider the exponentially transformed problem $\inf_{h \in \mathcal{A}(T)} \tilde{J}(x,t,h)$ where
\begin{equation}\label{eq_JDRSAM_criterion_tilde_J}
    \tilde{J}(x,t,h) := \ln\mathbf{E}\left[e^{-\theta \ln V(t,x,h)} \right]
\end{equation}
Note that because the term $e^{-\theta \ln V(t,x,h)}$ is bounded from below by 0, $\inf_{h \in \mathcal{A}(T)} \tilde{J}(x,t,h)$ is well defined which implies that there exists at least one minimizer $\tilde{h}$.
\begin{eqnarray}
        \mathbf{E}\left[e^{-\theta \ln V(t,x,h)} \right]
        = \mathbf{E}_{t,x}^{h}
        \left[ \exp \left\{ \theta \int_{t}^{T} g(s,X_s,h(s))
        ds -\theta \ln v
        \right\} \right]
                                                                                                                                \nonumber
\end{eqnarray}
(see for example Lemma 8.6.2. in~\cite{Ok03}) and hence
\begin{eqnarray}
        \inf_{h \in \mathcal{A}(T)}\mathbf{E}\left[e^{-\theta \ln V(t,x,h)} \right]
        &=& \inf_{h \in \mathcal{A}(T)}\mathbf{E}_{t,x}^{h}
        \left[ \exp \left\{ \theta \int_{t}^{T} g(s,X_s,h(s))
        ds -\theta \ln v
        \right\} \right]
                                                                                                                                \nonumber\\
        &=& I(v,x;h^{*}(t);t,T)
                                                                                                                                \nonumber
\end{eqnarray}
which proves that the optimal control $h^*(t)$ for the auxiliary problem $\sup_{h \in \mathcal{A}(T)}I(v,x;h;t,T)$ derived in Section~\ref{sec_Optim} is indeed optimal for the problem  $\sup_{h \in \mathcal{A}(T)}J(x,t,h)$.

\vbox{\hrule height0.6pt\hbox{%
   \vrule height1.3ex width0.6pt\hskip0.8ex
   \vrule width0.6pt}\hrule height0.6pt
  }

%
%

\end{document}